%% file: main.tex
\documentclass[11pt]{article}

\usepackage{commands}

\title{One Attack to Rule Them All: Tight Quadratic Bounds for Adaptive Queries on Cardinality Sketches}
\author{Edith Cohen\\
Google Research and Tel Aviv University\\
Mountain View, CA, USA\\
\texttt{edith@cohenwang.com}
\and
Jelani Nelson\\
UC Berkeley and Google Research\\
Berkeley, CA, USA\\
\texttt{minilek@alum.mit.edu}
\and
Tam\'{a}s Sarl\'{o}s \\
Google Research \\
Mountain View, CA, USA\\
\texttt{stamas@google.com}
\and
Mihir Singhal\\
UC Berkeley and Google Research\\
Berkeley, CA, USA\\
\texttt{mihir.a.singhal@gmail.com}
\and
Uri Stemmer\\
Tel Aviv University and Google Research\\
Tel Aviv-Yafo, Israel\\
\texttt{u@uri.co.il}
}
\date{}

\begin{document}
\maketitle
\begin{abstract}
   Cardinality sketches are compact data structures for representing sets or vectors. These sketches are space-efficient, typically requiring only logarithmic storage in the input size, and enable approximation of cardinality (or the number of nonzero entries). A crucial property in applications is \emph{composability}, meaning that the sketch of a union of sets can be computed from individual sketches. Existing designs provide strong statistical guarantees, ensuring that a randomly sampled sketching map remains robust for an exponential number of queries in terms of the sketch size $k$. However, these guarantees degrade to quadratic in 
$k$ when queries are \emph{adaptive}, meaning they depend on previous responses. 

   Prior works on statistical queries (Steinke and Ullman, 2015) and specific MinHash cardinality sketches (Ahmadian and Cohen, 2024) established that this is tight in that they can be compromised using a quadratic number of adaptive queries. In this work, we develop a universal attack framework that applies to broad classes of cardinality sketches. We show that any union-composable sketching map can be compromised with $\tilde{O}(k^4)$ adaptive queries and this improves to a tight bound of $\tilde{O}(k^2)$ for monotone maps (including MinHash, statistical queries, and Boolean linear maps). Similarly, any linear sketching map over the reals $\mathbb{R}$ and finite fields $\mathbb{F}_p$ can be compromised using $\tilde{O}(k^2)$ adaptive queries, which is optimal and  strengthens some of the recent results by~\citet{GribelyukLWYZ:FOCS2024}, who established a weaker polynomial bound.

\end{abstract}

\input{100-intro}

\input{200-prelims}

\input{300-pool}

\input{500-composable}

\input{600-linear}

\input{700-composable-prelims}

\input{800-termination-proofs}

\section*{Conclusion}
There are several potential directions for follow up improvements on our work. We established quadratic-size attacks for composable sketches that are monotone or with unique cores and for Boolean, $\mathbb{R}$ and $\mathbb{F}_p$ linear sketches. But the known attacks for integer linear sketches~\citep{GribelyukLWYZ:FOCS2024} and general composable sketches are super quadratic. We conjecture that a quadratic dependence is the right answer for any cardinality sketch.

We expect that our techniques and insights would lead to further progress on revealing the fundamental limitations of sketching methods under adaptive inputs. 
Specifically, we established that any composable map $S$ of subsets to a sketch of size at most $k$ has a small determining pool. That is, a set $L$ of size $\tilde{O}(k^2)$ (and $\tilde{O}(k)$ for monotone maps) so that $S(U)$ for random subsets $U$ is determined from $U\cap L$. This suggests that a similar attack paradigm can be designed against sketches for functions beyond cardinality.
%
We also conjecture that a universal attack of quadratic size exist on linear sketching maps for  
$\ell_2$ norm estimation and $\ell_2$-heavy hitters, extending the specialized attacks on the AMS sketch and Count Sketch~\cite{TrickingHashingTrick:AAAI2023}.

Universal attacks that are effective against broad families of sketching maps (and arbitrary estimators) establish vulnerabilities but also carry a promise of facilitating robustness. One potential approach is to apply attacks as part of an adversarial training process to make models more robust. This, by forcing a model to use random bits and effectively re-sample sketching maps for different inputs.

\section*{Acknowledgments}

Edith Cohen was partially supported by Israel Science Foundation (grant 1156/23). 
Uri Stemmer was Partially supported by the Israel Science Foundation (grant 1419/24) and the Blavatnik Family foundation.

\bibliographystyle{plainnat}
\bibliography{main,references,robustHH}
\input{400-analysis}

\input{2000-integrals}
\end{document}

%% file: 100-intro.tex
\section{Introduction}

Cardinality sketches provide compact, efficient summaries of sets, enabling efficient approximate cardinality estimation. Formally, a \emph{sketching map} 
$S$ maps a subset $U\subset \mathcal{U}$ to a succinct representation $S(U)$, where $|S(U)|\leq k$ is typically logarithmic in the input size $|U|$.
A key feature of cardinality sketches is composability, which allows essential set operations to be performed directly on the sketches:
Specifically, the sketch of  $U\cup \{x\}$ can be computed from the sketch of $U$ without additional access to $U$, and the sketch $S(U\cup V)$ can be computed from the sketches $S(U)$ and $S(V)$ of (potentially overlapping) subsets $U$ and $V$.
(see \cref{composablemap:def} for details). This property makes sketches well-suited for distributed and streaming data processing.


Cardinality sketches have been widely studied~\cite{FlajoletMartin85,flajolet2007hyperloglog,ECohen6f,ams99,BJKST:random02,KaneNW10,ECohenADS:TKDE2015,Blasiok20} and are extensively used in real-world applications~\citep{datasketches,bigquerydocs}.
They are deployed in network monitoring, data analytics, and large-scale tracking, where data is collected across multiple locations and time intervals. For example, they are used to track unique IP addresses, users, network flows, search queries, and resource accesses across distributed systems. Each location, such as a router, data center, or cache, 
maintaining a local sketch that summarizes the unique elements observed within a time window. 
These sketches are lightweight and are retained well beyond the massive raw data. They enable both local optimizations and large-scale analytics. Crucially, these applications depend on composability, which ensures that sketches can be seamlessly merged. In nearly all these real-world application, the setting is distributed (that is, not a single data stream) and in the so-called \emph{increment only} model (that is, there are no ``negative'' occurrences).

Notable cardinality sketching methods include the variety of MinHash sketches -- $k$-partition (PCSA)~\cite{FlajoletMartin85,flajolet2007hyperloglog},
 $k$-mins~\citep{ECohen6f,Broder:CPM00}, and bottom-$k$~\citep{Rosen1997a,ECohen6f,BRODER:sequences97,BJKST:random02} (see survey~\cite{MinHash:Enc2008,Cohen:PODS2023}) and linear sketches~\citep{CormodeDIM03,distinct_deletions_Ganguly:2007,KaneNW10}. Most prevalent in practice are implementations based on MinHash sketches~\cite{hyperloglog:2007,hyperloglogpractice:EDBT2013}.

Common to all methods is that they are parametrized by the sketch size $k$ and they specify a distribution from which a composable sketching map $S$, which maps sets to their sketches, is randomly sampled. 
The sampled sketching map must then be applied to all subsets, in order to facilitate composability.\footnote{and to support additional approximate queries in sketch space such as set similarity} The methods pair each sketching map $S$ with an \emph{estimator} that maps sketches $S(U)$ to estimates of the cardinality $|U|$.
These methods provide statistical guarantees on accuracy: for any particular subset $U$, the probability of a \emph{mistake} (relative error that exceeds the specification) over the sampling of the sketching and estimator map pairs decreases exponentially with the sketch size $k$. Therefore, with a sketch size of $k$, with high probability, an exponential number of queries in $k$ can be approximated with a small relative error with the same sampled sketching map. Importantly, this assumes that the queries are \emph{non-adaptive}, meaning 
they \emph{do not depend on the sampled sketching map} $S$.\footnote{for the purpose of analysis, the queries can be considered to be fixed in advance, before the map is sampled.}

\paragraph{Adaptive vs non-adaptive settings}
In adaptive settings, inputs are generated interactively, with each input potentially depending on previous outputs. These outputs -- and thus subsequent inputs -- may depend on the sketching map. Therefore, analyses that assume query independence from the map no longer hold. The vulnerability is that for any sketching map there are ``adversarial'' inputs with different cardinality but with the same sketch. The probability of randomly identifying such inputs without access to the map decreases exponentially with sketch size. But an adaptive interaction may leak information on the map which would allow for efficiently identify adversarial inputs.

Research works on the adaptive settings 
span multiple sub-areas including statistical queries~\citep{Freedman:1983,Ioannidis:2005,FreedmanParadox:2009,HardtUllman:FOCS2014,DworkFHPRR:STOC2015}, sketching and streaming algorithms~\citep{MironovNS:STOC2008,HardtW:STOC2013,BenEliezerJWY21,HassidimKMMS20,WoodruffZ21,AttiasCSS21,BEO21,DBLP:conf/icml/CohenLNSSS22,TrickingHashingTrick:arxiv2022,AhmadianCohen:ICML2024}, dynamic graph algorithms~\citep{ShiloachEven:JACM1981,AhnGM:SODA2012,gawrychowskiMW:ICALP2020,GutenbergPW:SODA2020,Wajc:STOC2020, BKMNSS22}, and machine learning~\citep{szegedy2013intriguing,goodfellow2014explaining,athalye2018synthesizing,papernot2017practical}.

\paragraph{Positive results via wrapper methods}
Robustness for adaptive queries can be obtained by \emph{wrapper} methods that use in a black box manner multiple independent copies of a randomized data structure that only provides statistical guarantees for non-adaptive queries. A simple
wrapper maintains $k$ copies and 
responds to each query using a fresh copy. This simple approach supports $k$ adaptive queries\footnote{in fact slightly less, since once has to apply a union bound to argue that all copies succeeded}. This can be improved to a number of queries that is \emph{quadratic} in the number of copies $k$. 
This was first established for adaptive statistical queries~\cite{Kearns1998} by
\citep{DworkFHPRR:STOC2015,BassilyNSSSU:sicomp2021}. The idea was to use differential privacy to protect the identity of the sampled keys. The quadratic boost followed from the advanced composition property of differential privacy~\citep{DMNS06} (that for a fixed privacy budget, allows for a number of privacy-preserving queries that is quadratic in the number of protected  units). The robustness wrapper  method of~\citet{HassidimKMMS20} lifted this framework from statistical queries to any randomized  data structure by interpreting the protected unit (that is, each sample) as the randomness initializing the data structure. An  alternative approach to obtain this quadratic boost without engaging the theory of differential  privacy was given by~\cite{Blanc:STOC2023} who established that it suffices to randomly sample a copy and use it to respond to each query. To summarize, these methods allow us to process $\tilde{O}(k^2)$ adaptive queries using super sketches of the form
$(S_1(U),S_2(U),\ldots,S_k(U))$ that 
use $k$ i.i.d.\ sketching maps $(S_i)_{i\in[k]}$. Observe that when the individual sketching maps are composable, so is the super-sketch.

\paragraph{Negative results by attack constructions}
Lower bounds were obtained using explicit constructions of attacks that specify an adaptive sequence of queries.
An attack on a family of sketching maps can be \emph{specific} to a particular estimator or it can be {\em universal} when it applies with {\em any} estimator.
\citet{HardtUllman:FOCS2014,SteinkeUllman:COLT2015} presented quadratic size universal attacks for 
adaptive statistical queries, based on Fingerprinting Codes~\citep{BonehShaw_fingerprinting:1998}. 
\citet{HardtW:STOC2013} designed a
polynomial-size universal attack on general linear sketches for $\ell_2$ norm estimation.
\citet{DBLP:conf/nips/CherapanamjeriN20} constructed an
$\tilde{O}(k)$ size specific
attack on the Johnson Lindenstrauss Transform with the standard estimator.
\citet{BenEliezerJWY21} presented an
$\tilde{O}(k)$ size specific attack on the AMS sketch~\citep{ams99} with the standard estimator.
\citet{DBLP:conf/icml/CohenLNSSS22}
presented an
$\tilde{O}(k)$ size specific attack on Count-Sketch~\citep{CharikarCFC:2002} with the standard estimator map.
\citet{TrickingHashingTrick:arxiv2022} constructed an  $\tilde{O}(k^2)$ universal attack on the AMS sketch and on Count-Sketch (when used for heavy hitter or inner product queries).

Specifically for cardinality sketches, 
\citet{DBLP:journals/icl/ReviriegoT20} and~\citet{cryptoeprint:2021/1139} constructed $\tilde{O}(k)$ size specific attacks on  the popular HLL sketch~\cite{hyperloglog:2007} with the standard estimator~\citep{hyperloglogpractice:EDBT2013}.
\citet{AhmadianCohen:ICML2024} constructed single-batch
$\tilde{O}(k)$ size specific attack with the standard estimator and
$\tilde{O}(k^2)$ universal attack for 
MinHash sketching maps. 
Finally,  a recent work by~\citet{GribelyukLWYZ:FOCS2024} on linear sketches construct universal attacks of a polynomial size $O(\mathrm{poly}(k))$ over the reals $\mathbbm{R}$, of $O(k^8)$ size over the integers $\mathbbm{Z}$, and of $O(k^3)$ size over a finite field $\mathbbm{F}_p$.

To summarize, existing cardinality sketch constructions  are guaranteed to be robust for a number of adaptive queries that is at most quadratic in sketch size. Existing attack constructions are either of super-quadratic polynomial size or narrowly apply to a particular sketch design. This raises the following natural questions:
\begin{ques} \label{main:problem}
    Can we devise a unified adaptive attack that  applies against any composable cardinality sketch? What is the broadest class of cardinality sketches for which we can design such an attack of polynomial size? Of quadratic size?
\end{ques}

\subsection{Overview of contributions}
Our results affirmatively settle Question~\ref{main:problem} for a broad classes of sketching maps, namely composable sketches and linear sketches. 

Our main contribution is an attack framework that unifies, simplifies, and generalizes prior works including
the method of~\cite{AhmadianCohen:ICML2024}, which was specifically tailored for certain known MinHash cardinality sketch designs, and the lower  bound  of \cite{UllmanSNSS:NEURIPS2018,HardtUllman:FOCS2014,SteinkeUllman:COLT2015} that was specific for adaptive statistical queries~\citep{Kearns1998} and based on fingerprinting codes~\citep{BonehShaw_fingerprinting:1998}.

A sketching map is specified by a pair $(S,\mathcal{U})$, where $\mathcal{U}$ is a ground set of keys of size $|\mc U|=n$ and $S:2^{\mathcal{U}}$ represents a map from sets to their sketches. 
Our interaction model, described in detail in \cref{prelim:sec}, follows prior work and corresponds to  typical applications in practice. An \emph{attacker} issues a sequence of adaptive queries $(U_i)$ where $U_i\in 2^{\mathcal{U}}$ is a subset of our ground set.  A \emph{query responder} (QR) receives the sketch $S(U_i)$ of the query and returns a response $Z_i$.\footnote{This query-based interaction model corresponds to the typical practical setting where estimates are used for localized analytics on disjoint data across time and servers and for global analytics. As for a single stream model with query points, the attack can be implemented for linear maps that support negative updates by using updates to turn the frequencies in $U_i$ into the frequencies in $U_{i+1}$. Otherwise, in the common scenario of no deletions, the count is monotone increasing and robustness can be secured with only a logarithmic blowup~\citep{BenEliezerJWY21}.}

Instead of approximate cardinality queries, we task the query responder with the simpler \emph{soft threshold} queries. For fixed values $A<B$, the response is a single bit $Z$ that is correct if $Z=0$ when the cardinality is  $|U|\leq A$ and $Z=1$ when $|U|\geq B$. Any response $Z\in\{0,1\}$ is acceptable as correct when $|U|\in (A,B)$. Observe that these simpler queries, and in particular larger ratios $B/A$, only makes the query responder more resilient to attacks. Our attacks apply with $B/A = \Omega(1)$.
The goal of the attack is to force the query response algorithm to make mistakes on at least some constant fraction $\eta$ of the queries.

\subsubsection{Attack Framework}

Our attack (see details in Section~\ref{unifiedattack:sec}) specifies a distribution $\mathcal{Q}$ over $2^\mathcal{U}$ where a rate $q$ is sampled from some distribution and a subset $U\sim \Bern[q]^{\mathcal{U}}$ is sampled by including each key independently with probability $q$. The actual query subset is $U\cup M$, where $M\subset\mathcal{U}$ is an initially empty set that we refer to as a \emph{mask}.  The query responder answers a soft threshold query based on the sketch $S(U\cup M)$. When the response is $1$, the attacker increments the score of all keys in $U\setminus M$. Keys in $\mathcal{U}\setminus M$ with scores that are sufficiently above the median score are added to $M$. 

The goal of the attack is to force an incorrect response on at least some fixed constant $\eta$ fraction of the specified number of $r$ queries.  The analysis idea is as follows: 
if the responder makes many intentional unforced errors, and thus does not reveal enough about the map $S$, then our goal is achieved. Otherwise, the attacker makes progress in building the mask $M$ and specifically towards making the distribution $\mathcal{Q} \cup M$\footnote{The notation $V\sim \mc Q \cup M$ means that we sample $U\sim \mc Q$ and take $V =U\cup M$.} 
\emph{adversarial to the sketching map} $S$. The distribution is $\eta$-adversarial if for \emph{any} estimator map $\phi$, $\phi(S(V))$ is incorrect with probability at least $\eta$. When the distribution becomes adversarial, the query responder, even if it knows $\mathcal{Q} \cup M$, must make forced errors on at least an $\eta$ fraction of the remaining queries.



Our analysis of the attack relies on the existence of a \emph{determining pool} $L\subset \mathcal{U}$, a smaller set of keys $|L|\ll |\mathcal{U}|$ 
such that we can determine the sketch of the query $U\cup M$ only from its intersection with the pool: 
\begin{definition} [Determining pool; simplified]
    For a sketching map $(S,\mathcal{U})$, a set $L\in 2^\mathcal{U}$ is a \emph{determining pool} if for any $M\subset L$, with probability at least $1-\delta$ over $U\sim \mathcal{Q}$, the sketch $S(U\cup M)$ can be computed from $M$ and $U\cap (L\setminus M)$. 
\end{definition}

\begin{theorem} [Informal; See Theorem~\ref{metacorrectness:thm} and Lemma~\ref{limitedpoolnaturalQR:lemma}] \label{informalmain:thm}
For any fixed $\eta<1/4$, if the sketching map $(S,\mathcal{U})$ has a determining pool $L\in 2^\mathcal{U}$ so that $|L| \ll n$,\footnote{$|L|<c n$ for some universal constant $c<1$}, then with probability $1-\delta$, our unified attack 
after $r=O(\log(1/\delta)|L|^2 \log n)$ queries
forces at least $\eta r$ mistakes.
This holds even against a powerful and strategic query response algorithm that is tailored to our query distribution at each step.
\end{theorem}

To establish that our attack works for a class of sketches $\mathcal{S}$, all we need to do is to prove that a determining pool of size $\ell \ll n$ exists for each sketching map in the class.
We demonstrate (see Section~\ref{examples:sec}) how the lower bound of ~\cite{SteinkeUllman:COLT2015} on statistical queries and the results of~\cite{AhmadianCohen:ICML2024} for MinHash sketches can be obtained by simply specifying the determining pool of the sketching map.
We also apply our framework to obtain new results for very broad classes of sketching maps: composable sketches and linear sketches.

\begin{remark} [Implication for Differentially Private Data Analysis]
    The existence of small determining pools for any composable sketch reveals inherent limitations of performing differentially private data analysis in the sketch space rather than directly on the original query sets. Specifically, the ``sensitivity'' depends inversely on the pool size $\tilde{O}(k)$ rather than on the potentially much larger query set size. This generalizes findings by~\citet{DLB:PET2019} for specific sketching maps.
\end{remark}

\subsubsection{Application to Composable Sketches}
(See details in Section~\ref{composablemaps:sec} and Sections~\ref{proofcomposable:sec}--\ref{proofgen:sec} for analysis)
A sketching map $(S,\mathcal{U})$ is \emph{composable} if 
for any two subsets $U,V\in 2^{\mathcal{U}}$, the sketch $S(U\cup V)$ can be composed solely from the sketches $S(U)$ and $S(V)$ (without access to $U$ or $V$).

We say that a composable map has \emph{rank} $k$ if the representation of each sketch uses at most $k$ bits. Alternatively, it suffices that for each sketch $\sigma$, the cardinality of a minimal subset with sketch $\sigma$ is at most $k$, that is, when $U$ is minimal such that $S(U)=\sigma$ (we say $U$ is a \emph{core} of $\sigma$), then $|U|\leq k$.

A composable sketching map is \emph{monotone} (simplified; see \cref{monotonecomposablemap:def}) if for any $U\subset V$, the cardinality of the cores of $S(U)$ is no larger than the cardinality of the cores of $S(V)$.
Monotonicity is a natural property which means that the sketching maps are ``efficient'' in that they are no less informative for larger sets than for smaller sets.
The sketching maps of statistical queries and of MinHash sketches (see \cref{examples:sec}) and Boolean linear sketches (see \cref{booleansketches:example}) are monotone and composable.

\begin{theorem} [Informal, see Theorem~\ref{composablepool:theorem}] \label{informalcomposable:thm}
    For any composable sketching map  $(S,\mathcal{U})$ of rank $k$ such that $n=\tilde{\Omega}(k^2)$, the attack succeeds using $\tilde{O}(k^4)$ queries. If the map is also monotone and $n=\tilde{\Omega}(k)$, the attack succeeds using $\tilde{O}(k^2)$ queries.
\end{theorem}

We prove the theorem by constructing a determining pool that is of size $O(k^2\log (k/\delta))$ for any composable sketching map and of size $O(k\log (k/\delta))$ for any monotone composable sketching map, and then applying \cref{informalmain:thm}. The pool construction is simple and is based on a \emph{peeling} process of the ground set $\mathcal{U}$, where in each step we peel a core of the sketch of the remaining keys that is contained in the remaining keys. That is, a minimal subset $A_i\subset \mathcal{U}\setminus \bigcup_{j<i} A_j$ so that $S(A_i)=S(\mathcal{U}\setminus \bigcup_{j<i} A_j)$. For maps with rank $k$, the size of each layer satisfies $|A_i|\leq k$ and we establish that the number of layers we need is $O(k\log (k/\delta))$ for general composable sketching maps and $O(\log (k/\delta))$ for monotone maps. We also (Section~\ref{tightnonmonotone:sec}) give an example of a (non-monotone) composable map of rank $k$  where $\Theta(k\log k)$ layers of size $\Theta(k)$ are necessary for a pool, and hence, our analysis for the peeling construction is tight. A direct implication of our results is that there is no composable sketching map (that is, a deterministic cardinality sketch) with guaranteed accuracy for all queries.

Observe that if the size of the ground set is $n\leq k$, we can represent $S(U)$ by $U$. This makes $(S,\mathcal{U})$ trivially composable of rank $k$. Furthermore, since $S$ is injective, it is resilient to attacks, as the goal of an attack is essentially to identify sets with very different cardinality that produce the same sketch. Consequently, attacks are feasible only when $n>k$, meaning that our attacks, which operate when $n=\tilde{O}(k)$, are nearly optimal in terms of the required ground set size.

\subsection{Application to Linear Sketches}
(see details in Section~\ref{linearsketches:sec})
With linear sketches over a field $\mathbbm{F}$, the input is represented as a vector $\boldsymbol{v}\in \mathbbm{F}^n$, the sketching map is specified by a \emph{sketching matrix} $A \in \mathbbm{F}^{n\times k}$ where $k\ll n$, and the sketch of $\boldsymbol{v}$ is the product $A \boldsymbol{v} \in \mathbbm{F}^k$.  
When linear sketches are used as cardinality sketches, the 
ground set are the entries $\mathcal{U}\equiv [n]$ and the query set corresponds to the nonzero entries of $\boldsymbol{v}$. The cardinality $\|\boldsymbol{v}\|_0$ is the number of nonzero entries, also called the $\ell_0$ norm. We establish the following:

\begin{theorem} [Informal, see Theorem~\ref{linearsketches:thm}] \label{linearinformal:thm}
Over a field $\mathbbm{F}$ that is $\mathbbm{R}$ (the real numbers) or $\mathbbm{F}_p$ (the finite field of prime order $p$, for any $p$),
there is a constant $C$ so that for any sketching matrix $A \in \mathbbm{F}^{k\times n}$, when $n> C\cdot k \log k$, there is an attack that succeeds using $O(k^2 \log^2 k \log n)$ queries.
\end{theorem}

Our results for $\mathbbm{R}$ are in 
the \emph{real RAM model}, which assumes that the sketch has
the \textit{exact} values in $A \boldsymbol{v}$.
We specify a parameter $\gamma_0(A)$ that depends on sub-determinants of $A$ and our attack is parametrized by an upper bound $\gamma> \gamma_0(A)$.

In the case of $\F_p$, our quadratic attack size improves upon the corresponding result in \cite{GribelyukLWYZ:FOCS2024}, which used a cubic attack size. In the case of $\R$, the quadratic attack size is also smaller than that of \cite{GribelyukLWYZ:FOCS2024} (which used an unspecified polynomial number of queries), though the entries of our query vector are exponentially far in magnitude (unless we make a further ``naturalness'' assumption on the query responder, see \cref{restrictednonzero:rem}).

Our quadratic bounds are \emph{tight} in that respective quadratic upper bounds are guaranteed in the same models, that is, with $\mathbbm{F}_p$ and specifically with $\mathbbm{R}$, even when the input values are unrestricted. To see this, recall the 
cardinality sketch of~\citep{CormodeDIM03} or folklore linear sketches for $\ell_0$ sampling or estimation (designed for non-adaptive queries) use a sketching matrix and estimators that only depend on whether measurement values (that is, the values in the sketch) are zero or not: each measurement simply tests if the nonzero entries in the measurement vector intersect with those in the input (see also Example~\ref{booleansketches:example}). The measurement vectors correspond to random sets of different sizes and have randomly chosen values. The main property needed is that the probability that a measurement value is $0$ when the intersection is not empty is unlikely. We can then apply the robustness wrapper method mentioned above~\cite{HassidimKMMS20} to $k$ independent random maps to obtain a sketch that supports a quadratic number of adaptive queries.

\begin{remark}[Magnitude of values] \label{restrictednonzero:rem} 
The query vectors used in our attack for $\mathbbm{R}$ have nonzero entries where the logarithm of the ratio between the largest and smallest nonzero magnitudes is $O(n \log \gamma)$. 
We propose an alternative specification of the nonzero values (see Section~\ref{naturallimited:sec}) where the logarithm of the ratio is
$O(\log (n \gamma))$, but this attack requires that the query responder is \emph{natural} (uses only components of the sketch that contain information on the cardinality --- see \cref{naturalQR:def} for a precise definition). 
For that, we establish a weaker property that is the existence of a \emph{limited} determining pool of size $O(k\log k)$ (A pool is limited if a sufficient statistic for $|U|$, but not necessarily the full sketch, can be computed from $(U\setminus M)\cap L$). We then establish a variant of Theorem~\ref{informalmain:thm} (see \cref{limitedpoolnaturalQR:lemma}) where the statement holds when the determining pool is limited but the query responder is natural.
Note that the known cardinality sketches are natural and therefore so are the respective robust algorithms obtained via the mentioned wrapper methods. 
\end{remark}

\subsubsection{Ideas in our Attacks on Linear Sketches}
Observe that the sketching map is a linear operator that is not composable, therefore, we can not directly apply Theorem~\ref{informalcomposable:thm} to establish \cref{linearinformal:thm}.\footnote{Except for boolean linear sketches, which are composable monotone maps -- see Example~\ref{booleansketches:example}.} Moreover, when specifying the attack, it is insufficient to specify the queries as sets because 
the same set $U\subset [n]$ has multiple representations, which are all the vectors with the set of nonzero entries being $U$, and there are multiple sketches of the same set. We therefore need to augment the description of the attack queries so that it specifies a vector, that is, values for the nonzero entries. 

Note that simply choosing the values of the nonzero entries to be $v_i=1$ will not work:
A sketching matrix that includes the measurement row vector 
$\boldsymbol{1}_n$ would have
the measurement value
$A \boldsymbol{v} = \mathbbm{1}_n\cdot \boldsymbol{v} = \|\boldsymbol{v}\|_0$, which is the exact cardinality value. Additionally, the smallest determining pool would be the full set $[n]$, because inclusion of any new key impacts the sketch. Therefore, for an attack that is effective against any sketching matrix, we need to specify these nonzero entries in our attack queries with some cleverness -- so that a determining pool of size $\tilde{O}(k)$ would exist. We do so using different and ad hoc approaches for the real RAM, for the real RAM with smaller entries (see Remark~\ref{restrictednonzero:rem}), and for the fields $\mathbbm{F}_p$.
There is, however, one common ingredient in all these constructions:
The determining pool $L$ is constructed from what we call a \emph{basis pool} of the column vectors $a^{(i)}\in \mathbbm{F}^k$ of the sketching matrix $A$. A basis pool
has the property that 
for any $M\subset [n]$, with probability at least $1-\delta$ over $U\sim \mathcal{Q}$, the span of the column vectors $(A_{\cdot j})_{j\in U\cup M}$ is equal to the span of the column vectors $(A_{\cdot j})_{j\in M \cup U\cap L}$. Surprisingly, the existence of such a basis pool of size $O(k\log k)$ falls out as a special case of the existence of pools for composable sketches (Theorem~\ref{informalcomposable:thm}) by observing that the span of a set of vectors in vector spaces is a monotone composable map of rank $k$. So after all, and indirectly, we do use our analysis for composable sketches.



\subsection{Statistical queries and MinHash Sketches} \label{examples:sec}
As a warm up and additional context for the reader, we recover prior results on statistical queries~\citep{SteinkeUllman:COLT2015} and MinHash cardinality sketches~\citep{AhmadianCohen:ICML2024} as applications of our framework.  We do so by demonstrating that the respective sketching maps are monotone composable. 

\paragraph{Statistical Queries}
The perhaps simplest composable cardinality sketches utilize a uniform random sample $R \subset \mathcal{U}$ of size $|R| = k$. 
The sketching map $S$, maps each set $U$ to its intersection with the sample, $S(U)= R \cap U$. This map is clearly composable since $S(U\cup V)= (U\cup V)\cap R = U\cap R \cup V\cap R = S(U) \cup S(V)$.
The map is also monotone: When $U\subset V$, $U\cap R \subset V\cap R$ and hence $|S(U)|\leq |S(V)|$. The determining pool for the sketching map, with $\delta=0$, is simply the set $R$, which has size $k$.

Cardinality queries are statistical queries and
the estimate $|R \cap U| \cdot n/k$ is a good approximation of $|U|$ when $U$ is sufficiently large, that is, $|U|=\Theta(|\mathcal{U}|)$. Observe that when the queries are non-adaptive (do not depend on the sample), we can accurately answer a number of queries that is exponential in the sample size $k$.  

This simple fixed-size sample sketch is 
effective when the queries have large cardinality. 
When the query set $U$ is small, the intersection with the sample would be small or even empty, making it impossible to reliably approximate $|U|$ to within a relative error. 
MinHash sketches~\citep{FlajoletMartin85,ECohen6f,flajolet2007hyperloglog,Rosen1997a,ECohen6f,Broder:CPM00,BJKST:random02}, which provide statistical guarantees for all cardinality values, can be viewed as drill-down samples instead of a fixed size sample. We sample a partial order of random priorities for all keys in $\mathcal{U}$. This determines the sketching map, roughly by selecting the $k$ top priority keys that are in the set according to the partial order. In this way, there is a size $k$ sketch for any subset and enough information also in the sketches of small sets to facilitate accurate estimates. We examine two MinHash sketch types, describe the  sketching maps, argue they are monotone composable, and  describe the determining pool. 

\paragraph{HyperLogLog (\texorpdfstring{$k$}{k}-partition) sketch}
A sketching map of $k$-partition sketches such as the HyperLogLog (HLL) sketch~\cite{flajolet2007hyperloglog} is specified by a (randomly chosen) partition of the keys in $\mathcal{U}$ to $k$ buckets and a (random) priority order on the keys in each bucket. For a fixed such partition and order, 
the sketch $S(U)$ of a subset $U$ depends on the lowest priority key in each of the buckets. This sketching map is composable since the sketch
$S(U\cup V)$ can be obtained by taking from each bucket, the highest priority key in the bucket among those in $S(U)$ and $S(V)$. The sketching map is monotone since the sketch size $S(U)$ is the number of buckets for which there is at least one representative from $U$. When $U\subset V$, $V$ has at least as many buckets represented as $U$.
The $\delta$-determining pool $L$ includes the $\log(k/\delta)/q_{\min}$ lowest priority keys from each of the $k$ buckets. Observe that a randomly chosen $U$ with rate at least $\qmin$ is $(1-\delta)$ likely to have its highest priority key in each of the buckets present in $L$. 



\paragraph{Bottom-\texorpdfstring{$k$}{k} sketch}
A sketching map of bottom-$k$ sketches~\citep{Rosen1997a,ECohen6f,BRODER:sequences97,BJKST:random02} is specified by priorities according to a random permutation of $\mathcal{U}$. For a fixed permutation, the sketch $S(U)$ depends on the $k$ highest priority keys in $U$.
The map is composable because the $k$ highest priority keys in $U\cup V$ are the $k$ highest priority keys in $S(U)\cup S(V)$.
The size of the sketch $S(U)$ is $\min\{|U|,k\}$ and the map is clearly monotone.
The determining pool $L$ includes the 
$k \log(k/\delta)/q_{\min}$ highest priority keys. It is easy to verify that the probability for a random $U$ selected with rate at least $\qmin$ that its $k$ highest priority keys are not included in the pool $L$ is at most $\delta$.

%% file: 200-prelims.tex
\section{Preliminaries} \label{prelim:sec}

\subsection{Notation}
We use the following notational conventions. 
\begin{itemize}
\item For a set $V$ and rate $q\in[0,1]$,
$\Bern[q]^{V}$ denotes the distribution over subsets that includes each 
$x\in V$ independently with probability $q$.
\item For $m\in \mathbbm{N}$, $\binom{V}{m}$ denotes the set of all subsets of $V$ of size exactly $m$.
\item For a set $V$, $\texttt{Uniform}[V]$ denotes the uniform probability distribution over $V$.
\item For a property $P$, we write $\mathbf{1}_P$ to denote the indicator function for $P$, which is equal to $1$ if $P$ holds and $0$ otherwise.
\end{itemize}

\subsection{Interaction Model}

We follow (and slightly extend) the model of~\cite{AhmadianCohen:ICML2024}.
Let the ground set of keys be denoted $\Uu$ and let $n$ be its size. We will denote the set to be sketched by $U \subset \Uu$. In some settings, there will be additional information provided to the sketch other than simply the set $U$. For example, in the case of linear sketches, where $U$ is the set of nonzero indices of the input vector, the vector contains the additional information of the values of the nonzero entries. We denote this additional information by $X \in \Xx$, where $\Xx$ is some arbitrary set which depends on the setting.

There are three interacting players, the \textbf{attacker} (which selects a query $(U,X)$), the \textbf{system} (which produces a sketch $S(U,X)$), and the \textbf{query responder (QR)} (which attempts to answer the threshold problem for the query).

We fix query thresholds $A < B$ in advance\footnote{We could also have the attacker specify $A$ and $B$ itself (subject to some restriction on $B/A$), since a QR that can give cardinality estimates would still be able to answer them. The attack that we provide, however, does not require $A$ and $B$ to change, so for simplicity we fix them in our model.}.
The interaction proceeds in steps as follows.
\begin{itemize}
\item
The \textbf{attacker} generates a query $(U,X)$, and sends $(U,X)$ to the system.
\item
The \textbf{system} receives $(U,X)$ from the attacker, computes the sketch $S(U,X)$, and sends $S(U,X)$ to the QR.
\item
The \textbf{query responder (QR)} then, having received the sketch, responds to the threshold problem for $|U|$, with either ``0'' or ``1.'' We say the response is \textit{correct} if it is 0 if $|U| \le A$ and 1 if $|U| \ge B$ (and if $A < |U| < B$ it is always considered correct).
\end{itemize}

Here, we allow the query responder to be adaptive over the course of multiple queries. That is, the QR can incorporate the information of the previous values of $S(U, X)$ when it is answering the current query\footnote{In fact, for the attack we present, the QR may even be allowed to access the actual previous values of $U$ and $X$. It can even access the entire internal state of the attacker, up to the distribution of queries that it chooses.}.




In the next section, we will provide a universal attack --- that is, a single attack algorithm that is effective against any collection of sketching maps that have a \textit{determining pool} (we will define this notion formally later). The attack will force the QR to make a constant fraction of errors during the attack with high probability.

%% file: 300-pool.tex
\section{Determining pools and the unified attack algorithm} \label{unifiedattack:sec}

\subsection{Universal attack algorithm} 

We will first describe the attack algorithm which we use, before stating the conditions under which it works.

The universal attack on a sketching map $S:\Uu\times\mathcal{X}$ is
described in \cref{metaattack:algo}. The attacker gets a 
ground set $\mathcal{U}$ of size $n$ and a family of (parametrized) distributions $\mathtt{Aux}[M_1, M_2]$ over $\mathcal{X}$; these are the distributions from which we will sample $X$. We will specify conditions on these distributions later.

Essentially, the goal of the algorithm will be to try to build up a mask $M$ which should end up containing most of the elements of a determining pool $L$. 

First we describe how the attacker samples query sets to give to the system. In short, we will do this by sampling a rate $q$ according to some fixed distribution, and then picking each element of the ground set with probability $q$, and then taking the union with a mask. We start by specifying the distribution $\nu$ of the rate.

\begin{definition}[Rate distribution] \label{ratedist:def}
Pick arbitrary $\qmin, q_1, q_2, \qmax$ such that
\begin{gather}
0 < \qmin < q_1 < \f {A-|L|}n, \label{q-bd-1:eq}\\
\f Bn < q_2 < \qmax < 1, \label{q-bd-2:eq}
\end{gather}
and moreover, each inequality is true by an additive difference of $\Omega(1)$. (Here, $|L|$ is the size of a determining pool\footnote{This requires $|L| < A - \Omega(n)$, which we state as an assumption in our theorem statement.}.) Then, let $\nu(q)$ be the probability density function which is proportional to $C_\nu f(q)/(q(1-q))$, where $f(q)$ is the following piecewise function, and $C_\nu$ is a constant selected to make the total probability density equal to 1.
\begin{equation} \label{f-def:eq}
f(q) = 
\begin{cases}
0, & q \le \qmin \\
(q - \qmin) / (q_1 - \qmin), & \qmin \le q \le q_1 \\
1, & q_1 \le q \le q_2 \\
(\qmax - q) / (\qmax - q_2), & q_2 \le q \le \qmax \\
0, & q \ge \qmax
\end{cases}
\end{equation}
A plot of $\nu(q)$ is shown in \cref{fig:nu}.
\end{definition}

\begin{figure}
\centering
\begin{tikzpicture}
  \begin{axis}[
    xlabel={$q$},
    ylabel={$\nu(q)$},
    domain=0:0.5,
    samples=200,
    width=10cm,
    height=7cm,
    axis x line=bottom,
    axis y line=left,
    grid=none,
    xtick={0.1, 0.2, 0.25, 0.35},
    xticklabels={$q_{\min}$, $q_1$, $q_2$, $q_{\max}$},
    ytick=\empty,
    ymin=0, ymax=6,
    enlarge y limits=upper,
  ]
    \addplot[
      blue,
      thick,
    ]
    { ( (x<=0.1 ? 0 : (x<=0.2 ? 10*(x-0.1) : (x<=0.25 ? 1 : (x<=0.35 ? 10*(0.35-x) : 0)))) ) / (x*(1-x)) };
  \end{axis}
\end{tikzpicture}

\caption{Example plot of the density function $\nu(q)$}
\label{fig:nu}
\end{figure}

\begin{definition}[Query distribution] \label{querydist:def}
The \textit{query distribution}, which we denote by $\Qq$ for the rest of this paper, is the distribution obtained by sampling $q \sim \nu$, and then sampling from $\Bern[q]^\Uu$ (that is, sampling each element of the ground set with probability $q$).
\end{definition}

\begin{remark}[Choice of $\qmin, q_1, q_2, \qmax$ when $\eps$ is small]
We will also later consider the regime where $B/A = (1 + \e)$ and $|L|/n < \e$, where $\e$ is small but constant. In this case, we will explicitly choose the following values:
\begin{gather*}
q_1 = A/n - 2\e, \\
q_2 = B/n + 2\e, \\
q_1 - \qmin = \qmax - q_2 = \sqrt \e.
\end{gather*}
Note that, for constant but sufficiently small (compared to $A/n$ and $1-B/n$) values of $\e$, these still satisfy the conditions of \cref{ratedist:def}. Also, in this case, the support of $f$ is contained in an interval of size $o_{\e\to0}(1)$ around $q_1$, so (as long as $q_1$ is bounded away from 0 and 1) the integral of $f(q)/(q(1-q))$ is just $(1+o_{\e\to0}(1))/(q_1(1-q_1))$ times the integral of $f$, which is just $\sqrt \e + O_{\e\to0}(\e)$. Thus, we have
\begin{equation} \label{eq:cnu}
C_\nu = (1+o_{\e\to0}(1)) \f{q_1(1-q_1)}{\sqrt \e}.
\end{equation}
\end{remark}

Now, for the algorithm, we initialize an empty mask $M$ and initialize counters $C[x]\gets 0$ for each member of the ground set.
In each step, we sample a subset $U \sim \mathcal{Q}$, and then sample auxiliary information $X\sim \mathtt{Aux}[M,U\setminus M]$.   

We then send the query $(U\cup M, X)$ and receive a response $Z$ from the QR. We increment by $Z$ the counters of all $x\in U\setminus M$.  Keys whose count is sufficiently above the median count of $\mathcal{U}\setminus M$ are then inserted into the mask $M$. After $r$ queries, the attack will have completed.

\begin{algorithm2e}[t]\caption{\small{\texttt{Attack on a sketching map $S:\mathcal{U}\times\mathcal{X}$}} \label{metaattack:algo}}
{\small
\DontPrintSemicolon
\KwIn{$r$ (upper bound on number of queries), $|L|$ (upper bound on size of a determining pool)}
\KwIn{A ground set $\mathcal{U}$ of $n$ keys; distributions $\mathtt{Aux}[M_1, M_2]$ over $\mathcal{X}$ \tcp*{parametrized by disjoint $M_1, M_2\in 2^{\mathcal{U}}$}}
\KwIn{thresholds $0 < A < B < n$}
Pick $\qmin, q_1, q_2, \qmax$ as in \cref{ratedist:def}.

$M\gets \emptyset$ \tcp*{Initialize mask}
\lForEach(\tcp*[f]{Initialize counts}){key $x\in \mathcal{U}$}{$C[x] \gets 0$}
\For{$t=1,\ldots, r$}
{
Sample rate $q$ with density $\nu(q)$ \tcp*{support $[\qmin,\qmax]$ and $\nu$ as in \cref{querydist:def}}
Sample $U\sim \Bern[q]^{\mathcal{U}}$ \tcp*{independently include each $x\in\mathcal{U}$ with probability $q$} 
Sample $X\sim\mathtt{Aux}[M,U\setminus M]$ \tcp*{auxiliary info}
\textbf{Send} $(U \cup M,X)$ to System \tcp*{System computes $\sigma := S(U\cup M,X)$ and sends $\sigma$ to QR}\;
\textbf{Receive} $Z$ from QR \tcp*{QR returns
$Z\in\{0,1\}$}
\ForEach(\tcp*[f]{score keys}){key $x\in U$}{$C[x] \gets C[x]+Z$\;
\If(\tcp*[f]{test if score is high}){$C[x] \geq \mathtt{median}(C[\mathcal{U}\setminus M]) + 16\sqrt{r \log(rn)}$}{$M\gets M\cup\{x\}$}}
}
}
\end{algorithm2e}


\subsection{Determining pools}

We now describe what it means for a sketch to have a \textit{determining pool}. We will end up showing that the attack in \cref{metaattack:algo} works against any sketch that has a determining pool. 

We will establish that when the sketching map $S$ has a \emph{determining pool} $L\subset\mathcal{U}$ that is of small size $|L|\ll |\mathcal{U}|$, then \cref{metaattack:algo} halts after a number of queries that is quadratic in $|L|$. Informally, $L$ is a determining pool if $W=U\cap L\setminus M$ subsumes the information that is available to the query responder from the sketch $S(U\cup M, X)$ on the query $U$ (except with probability at most $\delta$). 

\begin{definition} [Determining pool] \label{pool:def}
Fix a sketching map $S$, a ground set $\mathcal{U}$, and distributions $(\mathtt{Aux}[U_1,U_2])_{U_1,U_2\in 2^{\mathcal{U}}}$. 
For fixed $M$ and $U$, let $\mc D(M,U\setminus M)$ be the distribution of the sketch $S(U \cup M, X)$ for $X\sim \mathtt{Aux}[M,U\setminus M]$.
A subset $L\subset \mathcal{U}$ is a $\delta$-\emph{determining pool} if there are distributions $\mc D'(M,W)$ defined for for all
$M\subseteq L$, $W\subseteq L\setminus M$, satisfying the following property: for any $M\subseteq L$ and $q\in [\qmin,\qmax]$,
\begin{equation*}
\E_{U\sim \Bern[q]^{\mathcal{U}}}\left[ \left\| \mc D(M,U\setminus M) - \mc D'(M, U\cap(L\setminus M)) \right\|_{\text{TV}} \right]\leq \delta,
\end{equation*}
where the norm $\|\cdot\|_{\text{TV}}$ denotes the total variation (TV) distance between the two distributions.
\end{definition}

We make a few remarks on this definition:
\begin{remark} [Pool via tail bounds] \label{delta12pool:rem}
The above is an expectation bound on the total variation distance between the true distribution $\Dd(M, U \setminus M)$ and the approximate distribution $\Dd'(M, U \cap (L \setminus M))$. In our results here, we prove this expectation via a tail bound; namely, we show for some $\dl_1, \dl_2$ that
\begin{equation*}
\Pr_{U\sim \Bern[q]^{\mathcal{U}}}\left[ \left\| \mc D(M,U\setminus M) - \mc D'(M, U\cap(L\setminus M) \right\|_{\text{TV}} \geq \delta_2 \right]\leq \delta_1.
\end{equation*}
In this case, we will refer to $L$ as a \textit{$(\dl_1, \dl_2)$-determining pool}; note that a $(\dl_1, \dl_2)$-determining pool is also a $(\dl_1 + \dl_2)$-determining pool.
\end{remark}

\begin{remark}[Coupling interpretation of determining pool] \label{coupling:remark}
We may also interpret this definition of a determining pool in terms of coupling as follows: for any $M \subseteq L$ and rate $q \in [\qmin, \qmax]$, when sampling $U \sim \Bern[q]$ there exists a coupled random variable $S'$ such that $S'$ and $U \setminus M$ are conditionally independent, conditioned on $U \cap (L \setminus M)$, such that $S'$ and $S(U \cup M, X)$ differ with probability at most $\delta$.
\end{remark}

Observe that a determining pool always exists, in particular, the full ground set $L = \mathcal{U}$ is trivially a $(0,0)$-determining pool with $\mc D'=\mc D$. A pool would be useful and the attack can be effective, however, only when the pool size is sufficiently smaller than the ground set size $|L|\ll n$.
Additionally, we aim for $\delta \ll 1/r$, where $r$ is an upper bound on the number of queries. With this choice, with probability $1-o(1)$, for all the queries in our attack, the information in the sketch is subsumed by $U\cap L\setminus M$ (that is, $\mc D'(U\cap L\setminus M)$ behaves like $\mc D(M,U\setminus M)$).

We show that when there is a determining pool, the unified attack succeeds (see analysis in \cref{analysisattack:sec}):
\begin{theorem} [Effectiveness of attack] \label{metacorrectness:thm}
Fix a ground set $\Uu$ of size $n$, and distributions $\Aux[M_1, M_2]$ over $\Xx$. Furthermore, let the query thresholds be $A = an$ and $B = bn$ for constants $a, b$ such that $0 < a < b < n$. Suppose that the sketching map $S$ has a $\delta$-determining pool $L$ (which is unknown to the attacker) of size $|L| < \min(A, n/2) - \Omega(n)$ and with $\delta < n^{-3}$. Then, for \textit{any} QR algorithm, the attack of \cref{metaattack:algo} will cause it to make at least $\eta r$ errors (for a constant $\eta$) in $r = O(|L|^2 \log n)$ rounds, with probability $1-O(n^{-1})$.

Moreover, as $b/a \to 1$ (keeping $a, b$ bounded away from $0, 1$ by an absolute constant), and $|L|/n \to 0$, the value of the constant $\eta$ approaches $1/4$.
\end{theorem}

Note that the existence of a determining pool depends on the distributions $\Aux[U_1,U_2]$ that we pick for the auxiliary information. Thus, in general, in order to use this theorem to demonstrate the existence of an attack, we will need to specify how to pick the distributions $\Aux[U_1,U_2]$ to cause a determining pool to exist.
Moreover, in order for the attack to work against a class $\Ss$ of sketching maps (rather than just a single sketch map), it must be the case that the distribution $\Aux[U_1,U_2]$ does not depend on the choice of $S\in \Ss$. However, the determining pool $L$ may (and in our constructions does) depend on the map $S$, since the attack works without knowing $L$.

To apply the theorem for a sketching method, which defines a distribution over a class of sketching maps $\Ss$, it suffices to show that for
parameters $\delta_1,\delta_2,\ell$,
there exist distributions $\Aux[U_1,U_2]$ so that
for every map in the support $S \in \Ss$ and every ground set of size $n = \omega(\ell)$, a $(\delta_1,\delta_2)$-determining pool of size $\ell$ exists. Our applications here satisfy these requirements. Note that it is also possible to establish that ``most'' random maps and ground sets satisfy the condition and then obtain respective weaker results. 


\begin{remark}[Constant in error rate] \label{error-rate:rem}
It seems that $1/4$ (rather than $1/2$) is a fundamental barrier to our exact approach, even if we pick a different function $f$ in \cref{ratedist:def}. See the conclusion for more discussion on this.
\end{remark}

\begin{remark}[Adversarial distribution]
We say that a query distribution $\mc D$ is $\eta$-\textit{adversarial} for the map $S$ if for \emph{any} estimator map $\phi:\Sigma \to \{0,1\}$
\[
\E_{V\sim \mc D} \left[{\bs 1}_{|U|\leq A }\cdot  {\bs 1}_{\phi(S(V))=1} +  {\bs 1}_{|U|\geq B }\cdot {\bs 1}_{\phi(S(V))=0} \right] \geq \eta,
\]
that is, any query responder must incur an error rate at least $\eta$.
One can view the result of \cref{metaattack:algo} as the following: if the QR is required to, whenever possible, ensure that the probability of an error is at most $\eta$, then at the end of the attack, we end up with an $\eta$-adversarial distribution. Additionally, this adversarial distribution takes a nice form: it consists of the union of a mask $M$, with the distribution of \cref{querydist:def}.
Observe that per \cref{failureprob:rem}, for any failure probability $n^{-h}$, with $h > 0$, we can guarantee that the adversarial distribution is produced with probability $1-n^{-h}$ when increasing the attack size by a factor of $O(h)$.
\end{remark}

\begin{remark}[Failure probability] \label{failureprob:rem}
For simplicity, we have stated the result with failure probability $1/n$. It is possible to amplify this to $n^{-h}$ for constant $h$ by simply repeating the algorithm sequentially; however, this would worsen the constant $c'$ in the previous remark. In our analysis, however, it is possible to avoid worsening $c'$ by simply choosing different parameters in our algorithm (and instead only worsen the constant factor on $r$). To point out the specific places, one needs to assume $\delta < n^{-2-h}$ for the error of the determining pool, and also add a factor of $\sqrt h$ to the amount by which a count needs to exceed the median count in order to be added to the mask.
\end{remark}

\subsection{Limited pools and natural QR}\label{naturallimited:sec}

We now define a weaker requirement on the pool that when coupled with a restriction on QR that it is \emph{natural}, the  statement of Theorem~\ref{metacorrectness:thm} holds. 

A \emph{limited pool} is a set of keys $L$ so that 
$W = U\cap (L\setminus M)$ contains all information available from the sketch on the rate $q$ (except with probability $\delta$). Note that this is a weaker requirement than Definition~\ref{pool:def} because $W$ may not subsume additional information on $U$ that is present in the sketch and is not relevant for the rate but may still be used by a strategic QR algorithm.


First, we define the notion of a \textit{sufficient statistic}:
\begin{definition}[Sufficient statistic for the rate $q$] \label{suffstat:def}
    Fix a sketching map $S$, a ground set $\mathcal{U}$, and distributions $(\mathtt{Aux}[U_1,U_2])_{U_1,U_2\in 2^{\mathcal{U}}}$. 
    For a fixed subset $M\subset \mathcal{U}$ and rate $q$, 
    consider the random variable 
    $S(U\cup M, X)$ where $U \sim \Bern[q]^{\mathcal{U}}$ and $X\sim \mathtt{Aux}[M,U\setminus M]$.
    Then $T_M(S)$ is a \emph{sufficient statistic} for the rate $q$ if 
    the conditional probability distribution of $S$, given the statistic $t = T_M(S)$, does not depend on $q$.
\end{definition}

\begin{definition} [Limited determining pool] \label{limitedpool:def}
    Fix a sketching map $S$, a ground set $\mathcal{U}$, and distributions $(\mathtt{Aux}[U_1,U_2])_{U_1,U_2\in 2^{\mathcal{U}}}$. 
    Let $T_M(S)$ be a sufficient statistic for $q$ as in \cref{suffstat:def}.
    
    For $M,U$. let $\mc D(M,U\setminus M)$ be the distribution of $T_M(S(U \cup M, X))$ for $X\sim \mathtt{Aux}[M,U\setminus M]$.

    A subset $L\subset \mathcal{U}$ is a \emph{limited} $\dl$-\emph{determining pool} 
  if for any
    $M\subseteq L$ and $W\subseteq L\setminus M$, there exist distributions $\mc D'(M,W)$, such that for any $M$ and $q\in [\qmin,\qmax]$ it holds that:
    \[
    \E_{U\sim \Bern[q]^{\mathcal{U}}}\left[ \left\| \mc D(M,U) - \mc D'(M, U\cap(L\setminus M) \right\|_{\text{TV}} \right]\leq \delta\ .
    \]
    We also define a limited $(\dl_1, \dl_2)$-determining pool similarly to \cref{delta12pool:rem}.
\end{definition}

\begin{definition} [Natural query responder] \label{naturalQR:def}
  We say that the QR is \emph{natural} if its actions, including choice of estimator maps, only depends on a sufficient statistic $T_M(S)$ for the rate $q$.   
\end{definition}
The standard estimators for cardinality sketches that are designed for non-adaptive queries are natural -- they use 
estimator maps that only depend on a sufficient statistic for the cardinality.
The robust QR algorithms obtained by wrapping these sketches with natural estimators using the robustness wrappers discussed in the introduction are therefore also natural. 

We establish the following (see analysis in \cref{analysisattack:sec}):
\begin{lemma} [Attack with limited pools and natural QR] \label{limitedpoolnaturalQR:lemma}
The statement of Theorem~\ref{metacorrectness:thm} holds when $L$ is a limited pool if QR is natural.
\end{lemma}

%% file: 500-composable.tex
\section{Composable Maps} \label{composablemaps:sec}

\begin{definition} [Composable Map] \label{composablemap:def}
Let $(S,\mathcal{U})$ be such that $\mathcal{U}$ is a finite set (which we call a \emph{ground set}) and
$S: 2^{\mathcal{U}} \to \Sigma$ is a map of 
subsets $U\in 2^{\mathcal{U}}$ to their sketches $S(U)\in \Sigma$.
We say that $(S,\mathcal{U})$ is \emph{composable} if there is a binary composition operation $\oplus$ over $\Sigma$ such that for any two subsets $U, V \in 2^{\mathcal{U}}$, $S(U\cup V) = S(U)\oplus S(V)$.
\end{definition}

\begin{definition} [Cores of a sketch]
    A subset $U\in 2^{\mathcal{U}}$ is a 
\emph{core} of a sketch $\sigma\in S(2^{\mathcal{U}})$ if it is a minimal subset for which $S(U)=\sigma$. That is, for any proper subset $U'\subset U$ it holds that $S(U')\not=\sigma$.
For a sketch $\sigma\in\Sigma$, we denote by $\mathtt{Cores}(\sigma)$ the set of all the cores of $\sigma$. For a subset $U$ we denote by $\InCores(U) = \{C\in\Cores(S(U)) \mid C\subset U \}$ the set of all cores of $S(U)$ that are contained in $U$.
\end{definition}
Observe that any sketch $\sigma$ can alternatively be represented by any one of its cores $C\in\Cores(\sigma)$.

We also consider the following subclass of composable maps:
\begin{definition} [Monotone Composable Map] \label{monotonecomposablemap:def}
A composable map $(S,\mathcal{U})$ is \emph{monotone} if the following 
holds for all $A, U_1, U_2\in 2^{\mathcal{U}}$ such that $U_1\subset U_2$:
When  $C_i\subset U_i$ ($i\in \{1,2\}$) are minimal for which  $S(C_i\cup A) =S(U_i \cup A)$, then $|C_1| \leq |C_2|$.
\end{definition}
Monotonicity is a natural property that  means that the cardinality of minimal subsets that ``determine'' the sketch of $U$ is monotone in $U$. 
The cardinality sketches reviewed in Section~\ref{examples:sec} are monotone. 

We define the \emph{rank} of a composable map $(S,\mathcal{U})$ to be the maximum possible size of a core:
\begin{equation}
\rank((S,\mathcal{U})) :=   \max_{\sigma\in S(2^{\mathcal{U})})} \max_{U\in\Cores(\sigma)} |U|\ .
\end{equation}

We show that when the sketch representation is limited to $k$ bits then
the rank is at most $k$:
\begin{lemma} [see Lemma~\ref{coresizeklimit:claim} in the sequel] \label{coresize:lemma}
Let $(S,\mathcal{U})$ be a composable map such that $|S(U)| \leq k$ for all $U\in 2^{\mathcal{U}}$, where for a sketch $\sigma$, $|\sigma|$ is the number of bits in the representation of $\sigma$. Then it holds that
$\rank((S,\mathcal{U}))\leq k$.
\end{lemma}

We establish the following:
\begin{theorem} [Pool for Composable Maps] \label{composablepool:theorem}
Let $(S,\mathcal{U})$ be a composable map of rank $k$. 
Let $\qmin > 0$. Then there is $L\in 2^{\mathcal{U}}$ such that
\begin{equation} \label{composablestatement:eq}
    \forall  q\geq \qmin, \forall M\in 2^{\mathcal{U}},\, 
\Pr_{U\sim \Bern[q]^{\mathcal{U}}}[S((U\cap L)\cup M)=S(U\cup M)] \geq 1-\delta\ 
\end{equation}
that is of size  
$|L| = O(k^2 \log(k/\delta)/\qmin)$. If in addition, $(S,\mathcal{U})$ is monotone, then there is such $L$ of size $|L| = O(k \log(k/\delta)/\qmin)$. The set $L$ is a $(\delta,0)$-determining pool (see \cref{delta12pool:rem}).
\end{theorem}

As an immediate corollary of Theorem~\ref{composablepool:theorem} and \cref{coresize:lemma}, we obtain:
\begin{corq}
The statement of \cref{composablepool:theorem} holds if instead of rank $k$ we require that the sketch size is bounded by $k$.
\end{corq}


It is immediate from \cref{pool:def} (see \cref{delta12pool:rem}) that $L$ that satisfies \eqref{composablestatement:eq} is a $(\delta,0)$-determining pool:
Observe that we do not need to specify the auxiliary information $\Aux[M, U\setminus M]$ in Algorithm~\ref{metaattack:algo} since there is a unique representation $S(U)$ for the sketch of each set $U$. We can effectively take $\Aux[U_1,U_2]$ to be a fixed point mass and hence $\mc D(M,U\setminus M) \equiv S(M\cup U)$ is a point distribution on the sketch function. We then define $\mc D'(M,W) := S(M,W)$ and obtain from \eqref{composablestatement:eq} that
$\mc D(M,Q\setminus M) = \mc D'(M,Q\cap (L\setminus M)$ with probability $\delta$. This satisfies \cref{pool:def} with $\delta_1=\delta$ and $\delta_2=0$.

We prove the theorem by specifying an explicit (and simple!) construction of the pool $L$ in Section~\ref{poolconstruction:sec} and then establish that \eqref{composablestatement:eq} is satisfied. The analysis is included in Sections~\ref{proofcomposable:sec}, \ref{uniqcores:sec}, and~\ref{proofgen:sec} with a roadmap provided in Section~\ref{roadmap:sec}.

\subsection{Pool Construction} \label{poolconstruction:sec}

\begin{definition} [Core Peeling]
A \emph{Core Peeling} $(A_i)$ of a sketching map $(S,\mathcal{U})$ is a sequence of subsets 
    \begin{align*}
        A_1 &\gets C\in \InCores(\mathcal{U}) &\\
        A_{i+1} &\gets C \in \InCores(\mathcal{U}\setminus \bigcup_{j\leq i} A_j ) &\, ; \, \text{ for } i\geq 1
    \end{align*}
\end{definition}
Observe that the peeled cores $(A_i)$ are disjoint.
The process stops when the ground set is exhausted (in which case $(A_i)$ is a partition) or when 
$S(\mathcal{U}\setminus \bigcup_{j\leq i} A_j) = S(\emptyset)$, which means that the keys $\mathcal{U}\setminus \bigcup_{j\leq i} A_j$ are transparent to the sketching map.
Otherwise, there is a nonempty $A_{i}$ since by the definition of cores, for any $V$, there must be a core of $S(V)$ that is a subset of $V$. 

Our constructed pool will be a prefix of the core peeling, that it, have the form
$L= \bigcup_{i=1}^\ell A_i$ for some $\ell$.

\subsection{Examples of Maps}

We list some examples of maps with different properties.
\begin{example} [MinHash sketching maps]
The MinHash based sketch maps and the sample based map with parameter $k$ described in Section~\ref{examples:sec} are composable and monotone (satisfy Definition~\ref{monotonecomposablemap:def}) and have rank $k$. Each sketch $\sigma$ has a single core, $\core(\sigma)$ and the sketch representation is its core $S(U):=\core(U)$.
The sample sketch (statistical queries), for a sample $R$ of size $|R|=k$, has $\core(S(U)) = U\cap R$. The core peeling has one layer $A_1=R$ with the remaining keys being transparent.
With bottom-$k$ sketches, $\core(S(U))$ contains the $k$ keys of lowest priority in $U$ or is equal to $U$ when $|U|\leq k$. Therefore,   $|S(U)|=\min\{k,|U|\}$.  The core peeling has $A_i$ as the set of keys with priorities $ki+1$ to $k(i+1)-1$.
\end{example}

\begin{example} [Vector Spaces] \label{vectorspace:example}
    Let $(S,\mathcal{U})$ be such that $\mathcal{U}$ are vectors in some vector space and $S(U)$ is a representation of $\spn(U)$ (the subspace spanned by $U$). We show that this is a monotone composable map of rank $k = \dim \spn(U)$.
    
    For a set $U$, $\Cores(S(U))$ are the basis sets of $\spn(U)$ that are in $\mathcal{U}$.  Since basis have size at most $k$, $\rank(S,\mathcal{U})\leq k$.
    \begin{itemize}
        \item {\bf Composability:} $\spn(U\cup V)$ is specified by $\spn(U)$ and $\spn(V)$ (take any basis of $B_U$ of $\spn(U)$ and basis $B_V$ of $\spn(V)$ then $\spn(U\cup V) = \spn(B_u\cup B_v)$.
        \item {\bf Monotonicity:} We first recall a basic fact for  sets of vectors $A,U$
        \[
        \min\{|B\cap U| \mid B\subset A\cup U \text{ is a basis of } \spn(A\cup U)\} = \dim(\spn(A\cup U))-\dim\spn(A)).
        \]
        To establish this fact consider such a basis $B$. It must hold that $|B\cap A|\leq \dim(\spn(A))$ and hence $|B\cap U| \geq |B|-|B\cap A| \geq \dim(\spn(A\cup U))-\dim\spn(A))$. To establish equality, let $B_A\subset A$ be a basis of $\spn(A)$ and construct $B_U\subset U$ such that $B_A\cup B_U$ is a basis of $\spn(A\cup U)$. Then, since $|B_A|=\dim(\spn(A))$, we must have $|B_U|=\dim(\spn(A\cup U))-\dim\spn(A))$.
        
        Now for monotonicity, consider sets of vectors $A$ and $U\subset V$. 
        It follows that $U\cup A \subset V\cup A$ and hence $\spn(U\cup A) \subset \spn(V\cup U)$ and hence $\dim(\spn(U\cup A)) \leq \dim(\spn(V\cup A))$ and hence
    $\dim(\spn(U\cup A))- \dim(\spn(A)) \leq \dim(\spn(V\cup A))-\dim(\spn(A))$. This combined with the fact establishes monotonicity.
    \end{itemize}
    
    A core peeling is simply \emph{basis peeling}, where each layer is a basis of the remaining vectors, selected from them. Note that $\InCores(U)$ are all the basis sets that are in the set $U$ and this is always a nonempty set. So the core peeling forms a partition of $\Uu$.
\end{example}

\begin{example} [Composable map where $\Cores(\sigma)$ have different sizes and is not monotone]
Let $x_1,x_2,x_3,x_4,x_5,x_6,x_7\in \mathcal{U}$ and define $S(U)=\text{``full''}$  if
$\{x_1,x_2,x_3\}\subset  U$ or $\{x_4,x_5,x_6,x_7\}\subset U$ and otherwise $S(U)=U\cap \{x_i\}_{i\in [7]}$. Then $\Cores(\text{``full''})=\{\{x_1,x_2,x_3\},\{x_4,x_5,x_6,x_7\}\}$, that is, contains two cores of different sizes $3$ and $4$. This composable map is also not monotone: Let $U_2=\{x_i\}_{i\in[7]}$ and $U_1=\{x_1,x_3,x_4,x_5,x_6\}$. Then $U_1\subset U_2$. The cores sizes of $S(U_2)$ are $3,4$ but $\Cores(S(U_1))=\{U_1\}$ that is of size 5.
\end{example}

\begin{example} [Composable map where sketches have exponential in $k$ many cores]
    We specify a composable sketching map where 
$|\Cores(\sigma)|$ is exponentially large in $k$. 
Let $k'=k/\log n$, and the map $S$ so that $S(U)=U$ when $|U|\leq k'$ and $S(U)=\text{``above-$k$''}$ when $|U|\geq k'$. The cores are 
$\Cores(S(U))=\{U\}$ when $|U|\leq k$ and 
$\Cores(S(U))=\binom{\mathcal{U}}{k'}$ (all distinct subsets of size $k'$ of $\mathcal{U}$) for $|U|\geq k'$. Similarly,
$\InCores(U) = \binom{U}{k'}$ when $|U|\geq k'$ and $\InCores(U)=\{U\}$ when $|U|\leq k$. A peeling consists of 
$\lfloor n/k' \rfloor$ disjoint sets of size $k'$ (and possibly one set of size $n \mod k'$) that cover the ground set. 
\end{example}

\begin{example} [Linear Sketches]
 Linear sketches do not satisfy Definition~\ref{composablemap:def} since they allow for multiple sketch representations and sketches of the same set (all vectors with the same set of nonzero entries represent the same set). Linear sketches are linear operators (instead of union-composable) over frequency vectors (instead of sets). That is, the sketching map satisfies $S(\boldsymbol{u}+\boldsymbol{v})=S(\boldsymbol{u})\oplus S(\boldsymbol{v})$: The sketch of the sum of frequency vectors is a sum of their sketches. 
\end{example} 

\begin{example} [Boolean Linear Sketches] \label{booleansketches:example}
Boolean linear sketches use boolean entry values and the logical $(\land,\lor)$ operation instead of multiplication and addition.
The ground set is $\mc U := [n]$ and a 
set $U\subset [n]$ is represented by the boolean vector $\bf v \in \{0,1\}^n$ where $v_i = \mathbf{1}_{i\in U}$. The sketching map is defined by a boolean matrix of the form $A \in \{0,1\}^{k\times n}$ and $S(U) = A \bf v$.
Boolean linear sketches are monotone and composable maps.
Clearly $S(U\cup V) :=S(U) \lor S(V)$, where $\lor$ is performed entry wise.
To establish monotonicity observe that the cores of a sketch correspond to minimal set covers (here rows with sketch entry ``$1$'' are the sets).
When linear sketches are used with estimators that only consider the sparsity structure of the sketch, as in \cite{CormodeDIM03}, they are effectively used as Boolean linear sketches. 
\end{example} 

Finally, we mention the
streaming cardinality sketch of~\cite{Vinod:ESA2022}. This sketch 
is limited to streaming and the sketch of a stream depends on the order and not just the subset. This sketch supports streaming updates (incremental) but not set union and is not composable.

\subsection{Roadmap for the Proof of Theorem~\ref{composablepool:theorem}} \label{roadmap:sec}
When $\rank((S,\mathcal{U}))\leq k$, it holds that  
$|A_i|\leq k$ for all layers. We establish that $\ell$ layers suffice for the pool and hence the pool size is at most $\ell\cdot k$.

\begin{lemma} [Number of Layers] \label{composablepool:lemma}
Under the conditions of Theorem~\ref{composablepool:theorem},
for $\delta>0$, 
a prefix of a core peeling $L= \bigcup_{i=1}^\ell A_i$ satisfies \eqref{composablestatement:eq} with
$\ell = O(k\log(k/\delta))$. If in addition, $(S,\mathcal{U})$ is monotone, then $\ell=O(\log(k/\delta))$ suffices.
\end{lemma}

Section~\ref{proofcomposable:sec} contains
preliminaries, the deferred proof of Lemma~\ref{coresize:lemma}, and a set up for the proof of Lemma~\ref{composablepool:lemma} that states
a sufficient \emph{termination condition} for the prefix of the peeling to satisfy \eqref{composablestatement:eq} (be a determining pool). 
The proof of Lemma~\ref{composablepool:lemma} (and hence Theorem~\ref{composablepool:theorem}) then reduces to establishing that the termination condition holds with probability $1-\delta$ for the stated value of $\ell$. 

As a warmup, the termination proof for the special case when 
sketches have \emph{unique cores} is given in Section~\ref{uniqcores:sec}. The termination proof for \emph{general} composable sketching maps is given in Section~\ref{proofgen:sec}.

%% file: 600-linear.tex
\section{Linear Sketches} \label{linearsketches:sec}

With linear sketches over a field $\mathbb{F}$, the input is represented as a vector $\boldsymbol{v}\in\mathbb{F}^n$. The ground set is the set of indices $\mathcal{U}=[n]$, and the query subset is the set of nonzero entries in $\boldsymbol{v}$. The cardinality, $\| \boldsymbol{v} \|_0$ is the number of nonzero entries, also called the $\ell_0$ norm of the vector $\boldsymbol{v}$.
A sketching map is specified by a matrix $A\in \mathbb{F}^{k\times n}$ and the sketch of $\boldsymbol{v}$ is the matrix vector product 
\[S(\boldsymbol{v}) := A \boldsymbol{v}\ .\]

We present attacks of size $\tilde{O}(k^2)$ for the fields $\mathbbm{R}$ (Section~\ref{reals:sec}) and
$\mathbbm{F}_p$ (Section~\ref{Fp:sec}).  
Observe that there are many possible forms for the same subset $U$ (namely, all vectors with the same set of nonzero entries). Therefore, when applying our attack of
Algorithm~\ref{metaattack:algo}
we need to specify the distributions 
$\mathtt{Aux}[M,U\setminus M]$ that assign values to 
the entries $M$ and $U\setminus M$. That is, the auxiliary information $X \sim \mathtt{Aux}[M,U\setminus M]$ will be exactly the information of the nonzero entries of $\bs v$.

We select an assignment so that there is a $(\delta,\delta)$-determining pool of size $O(k\log(k/\delta)/\qmin)$. The attack effectiveness then follows from \cref{metacorrectness:thm}.

Our determining pools are constructed from \emph{basis pools} or the more specific \emph{greedy basis pools} of the column vectors $(a^{(i)})$ of the sketching matrix $A$. We will use the following lemmas that establish the existence of basis pools of size $O(k\log(k/\delta)/\qmin)$:
\begin{lemma} [Basis Pool] \label{BasisPool:lemma}
Let $\{\boldsymbol{v}_i\}_{i=1}^n$ be a set of vectors in a $k$-dimensional vector space. Let $\qmin > 0$.
There is 
$L \subset \{\boldsymbol{v}_i\}_{i=1}^n$ of size
$|L| = O(k\log(k/\delta)/q_{\min})$ such that for any $q\geq q_{\min}$ and $M\subset \{\boldsymbol{v}_i\}_{i=1}^n$ it holds that
\begin{equation} \label{BasisPool:eq}
\Pr_{U\sim \Bern[q]^{[n]}}\left[\spn(((\boldsymbol{v}_i)_{i\in U \cup M} ) = \spn(((\boldsymbol{v}_i)_{i\in (U\cap L) \cup M})) \right]\geq 1-\delta.
\end{equation}
We refer to $L$ as a $\delta$-\emph{basis pool}. 
\end{lemma}
\begin{proof}
    Per \cref{vectorspace:example}, The map $S(I) = \spn(\left((\boldsymbol{v}_i)_{i\in I} \right)$ that maps a set $I$ of indices to the span of the respective vectors is monotone and composable and is of rank $k$. The statement of the lemma follows from a direct application of Theorem~\ref{composablepool:theorem}.
\end{proof}

\begin{definition} [Greedy Basis]
The \emph{greedy basis} of a sequence $(\boldsymbol{v}_i)_{i=1}^n$ of vectors, 
\[
\mathtt{GreedyBasis}\left((\boldsymbol{v}_i)_{i=1}^n \right)
 := \{i \mid \boldsymbol{v}_i\not\in\spn\left(\{ \boldsymbol{v}_j\}_{j<i} \right)  \}\]
is a subset of the indices $[n]$ that are selected greedily by going over the sequence in order and selecting $i$ if the vector $\boldsymbol{v}_i$ is linearly independent of the vectors of previously selected indices. Equivalently, it is the lexicographically first list of indices that corresponds to a basis of $\bs v$.
\end{definition}

\begin{lemma} [Greedy Basis Pool] \label{GreedyBasisPool:lemma}
Let $(\boldsymbol{v}_i)_{i=1}^n$ be a sequence of vectors in $\mathbbm{R}^k$. 
Let $\qmin = \Omega(1)$.
There is 
$L \subset \{\boldsymbol{v}_i\}_{i=1}^n$ of size
$|L| = O(k\log(k/\delta))$ such that for any $q\geq q_{\min}$ and $M\subset \{\boldsymbol{v}_i\}_{i=1}^n$ it holds that
\begin{equation} \label{GreedyBasisPool:eq}
\Pr_{U\sim \Bern[q]^{[n]}}\left[\mathtt{GreedyBasis}\left((\boldsymbol{v}_i)_{i\in U \cup M} \right)\subset L\cap U \cup M \right]\geq 1-\delta\ .
\end{equation}
We refer to $L$ as a $\delta$-\emph{greedy-basis pool}.
\end{lemma}
\begin{proof}
    For a ground set $[n]$ and for a sequence $(\boldsymbol{v}_i)_{i=1}^n$ of vectors in a $k$ dimensional space, we define the map
    $S(I) = \grb((a^{(i)}_{i\in I})$ that maps a set $I\subset [n]$ of indices to the greedy basis of the respective vectors. We show that $(S,[n])$ is monotone and composable. Because the greedy basis is unique, the map has unique cores (see \cref{MUcores:def}), which are $\core(U) = \grb(a^{(i)}_{i\in U})$. Since a core is a basis, it has size at most the dimension $k$. Hence, the rank of the map is $k$. The statement follows by applying Theorem~\ref{composablepool:theorem}.

 To establish composability, we use the following claim
 \begin{claim} [Composability]
 For sets of indices $I_1, I_2 \subset [n]$,
     \[
       \grb((a^{(i)}_{i\in I_1\cup I_2}) \subset \grb((a^{(i)}_{i\in I_1}) \cup
       \grb((a^{(i)}_{i\in I_2})\ .
    \]
 \end{claim}
    \begin{proof}
    Consider the greedy construction of $\grb((a^{(i)}_{i\in I_1\cup I_2})$. At each point, the span of the selected vectors contains the respective spans of the greedy selections for each of $I_1$ or $I_2$. Consider to the contrary the first time we select an index $j\not\in \grb((a^{(i)}_{i\in I_1}) \cup   \grb((a^{(i)}_{i\in I_2})$ and suppose without loss of generality that $j\in I_1$. We must have that $a^{(j)}$ is independent from the already selected vectors in $I_1\cup I_2$. But then it must also be independent of the greedy selections for $I_1$, since the span is contained in that of the greedy selection for $I_1\cup I_2$. This is a contradiction.            
    \end{proof}
    
 To establish monotonicy, recall that the core size, which is the (greedy) basis size, is simply the dimension of the spanned subspace. When $U_1 \subset U_2$, then 
 $\grb(a^{(i)}_{i\in U_1}) \subset \grb(a^{(i)}_{i\in U_2})$ and it holds that  $\dim(\spn(U_1))\leq \dim(\spn(U_2))$.\footnote{We use the definition of monotonicity in \cref{MUcores:def} that is specialized for unique cores.}
\end{proof}

We will also need the following:
\begin{definition} [Change of basis matrix] \label{FB:def}
 Let $(a^{(i)})\in \mathbb{F}^k$  be column vectors and let $B = (b_1,\ldots,b_\ell)$ by the indices of a linearly independent subset of vectors.
We use the notation $F_B \in \mathbb{R}^{k\times k}$ for an invertible change of basis matrix such that $F_B a^{(b_j)} = e_j$ for all $j\in [\ell]$ (here $e_j$ is the $j$-th standard basis vector). Such a matrix always exists since the columns $(a^{(j)})_{j\in B}$ are  linearly independent. 
\end{definition}

\subsection{Sketches over \texorpdfstring{$\F_p$}{F\_p}} \label{Fp:sec}
We consider a sketching matrix $A$ and query vectors $\bs v$ over $\F_p$, the finite field of prime order $p$. 

For a query subset $U$ and mask $M$, we consider the vector $\bs v$ sampled as follows:

\begin{equation} \label{aux-Fp:eq}
\begin{cases}
    v_i= 0 & \, \text{if } i\not\in U \cup M,\\
    v_i\sim \texttt{Uniform}(\F_p) &\, \text{if } i\in U \cup M.
\end{cases}  
\end{equation}

Note that some of the $v_i$ may end up being 0 for $i \in U \cup M$, and thus not counted in $\norm{\bs v}_0$. We will ignore this issue for now, and allow the auxiliary information to have zero entries in $U \cup M$, and then show how to fix this.

\begin{theorem}[Attack on a Sketching Matrix over $\mathbbm{F}_p$]\label{linearsketchesFp:thm}
If the auxiliary distributions $\Aux[M,U\setminus M]$ are taken as in \eqref{aux-Fp:eq} (allowing for zeroes in the auxiliary information), then for any sketching matrix $A\in \mathbb{F}_p^{k\times n}$, there is a $\delta$-determining pool of size
$O(k\log(k/\delta))$.
\end{theorem}

Before proving this theorem, we will first show how to circumvent the issue of sampling zeroes for $v_i$.

\begin{cor}
For any linear sketch $A$ over $\F_p$, the attack of \cref{metaattack:algo} will cause any QR to make $\Omega(r)$ errors in estimating $\norm{\bs v}_0$ in $r = O(|L|^2 \log n)$ rounds, with probability $1-O(n^{-1})$. Here the algorithm is given thresholds $A', B'$ which differ from the actual thresholds $A, B$ of the cardinality estimation task. The assumptions on the thresholds $A$ and $B$ are the same as in \cref{metacorrectness:thm}, except that we require $B/n = b < \f{p-1}{p} n$ (instead of being less than $n$).
\end{cor}
\begin{proof}
Note that in \eqref{aux-Fp:eq} each entry $v_i$ for $i \in U \cup M$ is zero with probability $1/p$, so by a Chernoff bound, with probability $1-O(n^{-2})$, we have $\norm{\bs v}_0 = \f{p-1}{p} |U \cup M| + o(1)$ . Therefore, if we pick $A' = \f{p}{p-1} A - cn$ and $B' = \f{p}{p-1} B + cn$ for a sufficiently small absolute constant $c$, whenever the QR correctly answers for $\norm{\bs v}_0$, its response is also correct for $|U \cup M|$ with thresholds $A'$ and $B'$ (with probability $1-O(n^{-1})$). Thus, the result follows from \cref{metacorrectness:thm} and \cref{linearsketchesFp:thm}.
\end{proof}

\begin{proof}[Proof of \cref{linearsketchesFp:thm}]
We sample $U\sim \mathcal{Q}$ as in  \cref{querydist:def}) and get the query subset 
$V = U \cup M$. We specify the distribution of query vector $\bs v$ as follows, with the entries $V\subset [n]$ sampled independently and uniformly at random from $\F_p$:
\[
 \begin{cases}
    v_i= 0 & \, \text{if } i\not\in V,\\
    v_i\sim \texttt{Uniform}(\F_p) &\, \text{if } i\in V.
\end{cases}
\]

Let $a^{(i)}$ refer to the $i$-th column of $A$, and furthermore, for any subset $I \subset [n]$, define $A|_I$ to be the matrix whose columns are $a^{(i)}$ for $i \in I$ (in increasing order of $i$), and similarly define $v|_I$.

Let $L$ be a basis pool for the sequence of column vectors
$(a^{(i)})$, as in \cref{BasisPool:lemma}, with $\qmin$ as in \cref{querydist:def}.

The following claim concludes the proof of \cref{linearsketchesFp:thm}:
\begin{claim}
    $L$ is a $(\delta,0)$-determining pool (Definition~\ref{pool:def}, \cref{delta12pool:rem})
\end{claim}
\begin{proof}
 
To establish that $L$ is a determining pool we need to show that for $M$ and over the sampling of $U$, the sketch distribution $S(\bs v)$ is determined (with probability at least $1-\delta$, up to TV distance at most $0$) by $W = U \cap (L \setminus M)$.

We obtain that any basis of 
$A |_{W\cup M}$ is a basis for all the columns of
$A |_{U\cup M}$, except with probability at most $\delta$.
The basis of $A |_{W\cup M}$ only includes columns with indices $B\subset W\cup M$ and hence is determined by $W$, recalling that the mask $M$ is a given. We show that when this is the case, we can determine the sketch distribution within TV distance $0$.


Let the invertible change of matrix $F_B$ be as in \cref{FB:def} and let $A' = F_B A$ with respective column vectors $(a'^{(i)})$. Since $F_B$ is invertible, we can instead consider the distribution of $A'\bs v$.
 Note that we can rewrite the sketch as 
\[A'\bs v = \sum_{i \in B} v_i a'^{(i)} + \sum_{i \in V \setminus B} v_i a'^{(i)}.\]
Since we have $v_i \sim \texttt{Uniform}(\F_p)$, the first $\ell$ entries in the column vector of the first term are distributed as a uniformly random element of $\spn((a'^{(i)})_{i \in B})=\spn((e_i)_{i\in[|B|]}$. The last $k-\ell$ entries are zero.
Note that the second sum is also an element of this span, with the last $k-\ell$ entries being zero. It is also independent of the first sum, conditioned on $U$. Thus, for each $U$, $A'\bs v$ is (exactly) distributed as a uniformly random element of $\spn((a'^{(i)})_{i \in B})$, and this distribution is a function only of $W$. Thus, we can determine exactly the distribution of $S(v) = A\bs v = F_B^{-1} A' \bs v$ conditioned on $U,M$ from $W$. 
\end{proof}
\end{proof}

\subsection{\texorpdfstring{Sketches over \(\mathbb{R}\)}{Sketches over R}} \label{reals:sec}

\begin{definition}[The parameter $\gamma_0(A)$]
For a matrix $A\in \mathbb{R}^{k\times n}$, denote the columns of $A$ by $(a^{(i)})$. 
We define the parameter $\gamma_0(A)$ to be the largest possible magnitude of any entry that can be in $F_B A$, for any basis $B= (b_1,\ldots,b_\ell)$ of $(a^{(i)})$ (here the elements of $B_i$ must themselves also be columns of $A$).
\end{definition}

We use the real RAM model, that is, assume the representation of the sketch is exact. We establish the following:
\begin{theorem} [Attack on a Sketching Matrix over $\mathbbm{R}$]\label{linearsketches:thm}
For any parameter $\gamma$, there are distributions
$\Aux[M,U\setminus M]$ such that 
for any sketching matrix $A\in \mathbb{R}^{k\times n}$ such that $\gamma_0(A)\leq \gamma$,
there is a $\delta$-determining pool of size
$O(k\log(k/\delta))$.
\end{theorem}

We propose two methods to specify the query vectors $\bs v$ from $M$ and $U\setminus M$. The first method (see Section~\ref{unrestrictednonzero:sec}) uses values where the logarithm of the ratio of largest to smallest nonzero magnitudes is $\log\left(\frac{\max_i |v_i|}{\min_{i \mid v_i\neq 0} v_i} \right)= O(n \log(\gamma_0))$. For this method we establish the existence of a $(\delta,\delta)$-determining pool (Definition~\ref{pool:def}, \cref{delta12pool:rem}) of size $O(k\log(k/\delta))$. The second method we propose 
(see Section~\ref{restrictednonzero:sec}) uses a much smaller log ratio of
$\log\left(\frac{\max_i |v_i|}{\min_{i \mid v_i\neq 0} v_i} \right)= O(\log(n \gamma_0))$ but we only establish the existence of a $(\delta,\delta)$-\emph{limited} determining pool (Definition~\ref{limitedpool:def}) of size $O(k\log(k/\delta))$ and hence, per \cref{limitedpoolnaturalQR:lemma}, the attack is effective  only against a \emph{natural} QR (see Definition~\ref{naturalQR:def}).

\subsubsection{Larger nonzero magnitudes ratio} \label{unrestrictednonzero:sec}

We sample $U\sim \mathcal{Q}$ (see \cref{querydist:def}) and as in
Algorithm~\ref{metaattack:algo} get a query subset 
$V = U \cup M$.
We specify values for the entries of the query vector $\bsv$ based on the query subset $V$ as follows:
\[
 \begin{cases}
    v_i= 0 & \, \text{if } i\not\in V\\
    v_i\sim \Exp(\beta^{-i} ) &\, \text{if } i\in V
\end{cases}
\]
where $\Exp(b)$ denotes an exponential random variable with expectation $b$, and $\beta =C\cdot \gamma n \log(n/\delta)k/\delta$ for some large enough constant $C$. \ecmargincomment{Mihir, please check if this setting of $\beta$ is what you meant? }

The high level idea is that the presence of the lowest index nonzero entry in the intersection of the measurement row and the query hides the presence of other nonzero entries. If we select into the pool a prefix of the nonzero entries in the measurement, our randomly selected $U$ is likely to have one of them as its lowest index nonzero. We only need to be careful that the lowest index nonzeros in different measurements do not overlap, which means they can be ``cancelled out'' by linear combinations and expose the presence of lower index nonzero entries.

Let $a^{(i)}$ refer to the $i$-th column of $A$.
Applying Lemma~\ref{GreedyBasisPool:lemma} to the sequence of
column vectors $(a^{(i)})$
we obtain that there is a $\delta$-greedy basis pool $L$ of $O(k\log(k/\delta))$ indices such that 
\[
\Pr_{U\sim \mathcal{Q}}\left[
\grb((a^{(i)})_{i \in U\cup M})\subset (U\cap L) \cup M\ \right]\geq 1-\delta.
\]

\begin{claim}
    The greedy basis pool $L$ is a $(\delta,\delta)$-determining pool.
\end{claim}
\begin{proof}

Define the basis indices $B = (b_1, \dots, b_\ell) = \grb((a^{(i)})_{i \in U\cup M})$ (for some $\ell \le k$). With probability at least $1-\delta$,
\begin{equation} \label{goodUM:eq}
    B = \grb((a^{(i)})_{i \in U\cup M})\subset (U\cap L) \cup M
\end{equation}
 and hence
$B$ can be determined from $(U\cap L) \cup M = W\cup M$. 

Recall that the sketch is defined as $S(v) = A\bsv$. It remains to show 
that when $M$ and our sampled $U$ are such that \eqref{goodUM:eq} holds and when
given $W := U\cap (L\setminus M)$ (and knowing $M$), we can approximate the distribution of $S(\bsv)$ conditioned on $U,M$ to within TV distance $\delta$.


For any subset $I \subset [n]$, define $A|_I$ to be the matrix whose columns are $a^{(i)}$ for $i \in I$ (in increasing order of $i$), and similarly define $v|_I$. Note that we can rewrite the sketch as 
\[S(\bsv) = \sum_{i \in V} a^{(i)} v_i = (A|_V) \bsv|_V. \]


Let the invertible matrix $F_B$ be as defined earlier and let $A' = F_B A$; since $F_B$ is invertible and determinable from $B$, it suffices to show that we can approximate the distribution of $A'\bsv$ (instead of $A\bsv$). We denote the columns of $A'$ by $a'^{(i)}$. From our assumption, every entry in $A'$ is known to be at most $\gamma$ in magnitude. 


Now, since the columns of $A'$ are just an invertible linear transformation of their respective columns of $A$, $B$ is also the greedy basis of the columns of $A'$. Thus, for every $i$, $a'^{(i)} \in \spn (a'^{(b_{j})})_{j: b_j \le i}$. But note that $a'^{(b_{j})} = F_B a^{(b_{j})} = e_j$, so this means that the first nonzero entry in row $j$ of $A'$ is a $1$ in column $b_j$ for $j \le \ell$, and rows $j > \ell$ are all zero. 

Now, we can write
\[(A'\bsv)_j = v_{b_j} + \sum_{i > b_j} A'_{ji} v_i\]
Note also that the $i \leq b_j$ terms do not appear in the above sum, as their coefficient in $A'$ is zero. Now, each $v_i$ is either zero or distributed as $\Exp(\beta^{-i})$. By a union bound, with probability at least $1-\dl/2$, we have $|v_i| \le O(\log (n/\dl) \beta^{-i})$ for all $i \notin B$. We now condition on the values of $v_i$ for all $i \notin B$. Then, we have
\[(A'v)_j = v_{b_j} + c_j\]
for some constants $c_j$, where with probability $1-\dl/2$ we have $c_j \le O(\gamma n \log (n/\dl) \beta^{-b_j-1})$ for all $j$. Thus, assuming this holds, the entries of $(A'\bsv)_j$ are (conditionally) independent and each distributed as
\[(A'\bsv)_j \sim \Exp(\beta^{-b_j}) + O(\gamma n{\log (n/\dl)} \beta^{-b_j-1}).\]
This has total variation distance $O(\gamma n{\log (n/\dl)} / \beta) < \dl/2 k$ from $\Exp(\beta^{-b_j})$. Thus, overall, the vector $A'\bsv$ has (conditional) TV distance at most $\dl$ from the vector with independent entries distributed as $\Exp(\beta^{-b_j})$ (and whose last $k-\ell$ entries are 0).

Thus, overall, we have established that, conditioned on $v_i$ for $i \notin B$, with probability at least $1-\dl/2$, the distribution of $A'v$ (and thus $Av$) can be approximated within TV distance at most $\dl/2$. Removing this conditioning, the approximation has TV distance at most $\dl$ from the distribution of $Av$ (still conditioned on $U$).

\end{proof}

\subsubsection{Smaller nonzero magnitudes ratio} \label{restrictednonzero:sec}


Let $q_0 = \qmin/2$.
We sample the rate $q$ and $U\sim \mathcal{Q}$ (see \cref{querydist:def}) and as in
Algorithm~\ref{metaattack:algo} get a query subset 
$U \cup M$.
We specify the query vector $\boldsymbol{v}$ as follows.
 We set $v_i= 0$ for $i\not\in U\cup M$. 
Let $H\subset U$ be a subsample of $U$ where each key is sampled with probability $q_0/q$. Importantly, note that $H$ does not depend on $q$: The sampling process of $U$ is equivalent to sampling $H\sim \Bern[q_0]^\mathcal{U}$ and then sampling $U\setminus H$  from $\Bern[(q-q_0)/(1-q_0)]^{\mathcal{U}-H}$. 
Therefore, all information on the rate $q$ is contained in $|U\setminus (H\cup M)|$. That is, $|U\setminus (H\cup M)|$ is a sufficient statistic for the rate per \cref{suffstat:def}.

We refer to the keys $H$ as \emph{large} keys.
The values $v_i$ for $i\in U\cup M$ are set as follows. The large keys and those in $M$,  $i\in H\cup M$, are set to independent $v_i \sim \Exp(\beta)$ for $\beta := C\cdot n\gamma k/\delta$ for
some large enough constant $C$. \ecmargincomment{Check that}
For the remaining keys, $i\in U\setminus (H\cup M)$, which we call \emph{small} keys, we set $v_i=1$.

Let $L$ be a $\delta$-basis pool with $q_0$ 
for the columns $(a^{(i)})$ of $A$ as in Lemma~\ref{BasisPool:lemma}.

\begin{claim}
    The basis pool $L$ is a limited $(\delta,\delta)$-determining pool  (\cref{limitedpool:def}).
\end{claim}
\begin{proof}
We show that in fact nearly all the information on the rate $q$ is contained in  $(U\setminus (H\cup M))\cap L$. That is, the part $(U\setminus (H\cup M))\setminus L$ that is outside the pool $L$ can have very little impact (in terms of TV distance) on the sketch distribution.

We apply Lemma~\ref{BasisPool:lemma} with $M$ and $H\sim \Bern[q_0]^\mathcal{U}$, and obtain that with probability at least $1-\delta$, there is a basis $B$ with indices in $(H\cap L) \cup M$ for the columns $H\cup M$.

Moreover, from the termination condition~\cref{terminationcond:eq} 
(in the proof for composable sketches) we know that it also holds that
\[
\spn((\bsv_i)_{i\in A_{>\ell}}) \subset \spn((\bsv_i)_{i\in H\cap L} )\ .
\]
In particular, all the vectors in $U$ that are linearly independent of $H\cup M$ are included in $L\setminus M$ (recalling that the basis pool is $L=A_{\leq \ell}$).

We wish to show that, given $W := (U\setminus (H\cup M))\cap L$, we can approximate the distribution of the sketch $S(\bsv)$ conditioned on $U,M$.
We show that 
the presence or nonpresence of keys from $U\setminus (H\cup M)$ for which $a^{(i)}$ is linearly dependent on $((\bsv_i)_{i\in H\cup M})$ can only change the sketch distribution within total variation distance $\delta$. Since the linearly independent keys are all in $L\setminus M$, this would conclude the proof of the claim.

Let $F_B$ be the invertible basis changing matrix as in \cref{FB:def} and let $A' = F_B A$ with column vectors $(a'^{(i)})$ be as defined earlier.
Consider the distribution of  
\[A'\bs v = \sum_{i \in B} v_i a'^{(i)} + \sum_{i \in ((H\cup M) \setminus B} v_i a'^{(i)}  + \sum_{i \in (U\setminus (H\cup M))\setminus L } v_i a'^{(i)} + \sum_{i \in (U\setminus (H\cup M))\cap L } v_i a'^{(i)}.\]

The first term is the sum of the $\ell$ first standard basis vectors multiplied by the respective $v_i\sim \Exp(\beta)$.
That is, it has independent $\Exp(\beta)$ random variables as the first $\ell$ entries and value $0$ in the last $k-\ell$ entries. 

The columns in the second and third terms are linearly dependent on $B$. The sum column thus also has the last $k-\ell$ entries zero.
The keys in $U\setminus (H\cup M)$ in the third and fourth terms contribute magnitude that is at most $n\gamma$ to each of the first $\ell$ measurement values $A' \bsv$.
Hence the first $\ell$ entries have magnitude at most $n\gamma$.
Using a similar argument to that in \cref{unrestrictednonzero:sec}, the one-entry TV distance between $\Exp(\beta)$ and $\Exp(\beta)+c $ for 
$|c| \leq \gamma n$ is at most $\delta/k$. Therefore, the TV distance of the column vectors can be at most $\delta$.  

The column vector from the last term may have nonzeros in the entries $j>\ell$ but it depends only on $W\subset  U\cap (L\setminus M)$.
This concludes the proof of our claim that $L$ is a $(\delta,\delta)$-limited determining pool.

\end{proof}

%% file: 700-composable-prelims.tex
\section{Preliminaries on Composable Maps} \label{proofcomposable:sec}
This section contains preliminaries and a setup for the analysis.
We start with helpful notation and establish properties of composable maps.

For a sequence of sets $(X_i)$ we define
$X_{\leq i} := \bigcup_{j\leq i} X_j$ and similarly define
$X_{\geq i}$, $X_{<i}$.

Since set union is associative and commutative, so is the sketch composition operation ``$\oplus$.'' For sets $(U_i)$ we have that $\bigoplus_i S(U_i) = S(\bigcup_i U_i)$. In particular, for any sketch $\sigma$ it holds that
\begin{equation} \label{selfcompose:eq}
    \sigma\oplus\sigma = \sigma
\end{equation}
(this because $\sigma= S(U) =S(U\cup U)=S(U)\oplus S(U) = \sigma\oplus\sigma$.)

The following property states that if a subset is sandwiched between two subsets with sketch $\sigma$, it must also have sketch $\sigma$.
\begin{claim} [Midpoint Property] \label{midpoint:claim}
    Let the subsets $U_1 \subset U_2 \subset U_3 \in 2^{\mathcal{U}}$ be such that $S(U_1)= S(U_3)$.
    Then $S(U_2)=S(U_1)=S(U_3)$.
\end{claim}
\begin{proof}
Let $\sigma_1 = S(U_1)=S(U_3)$,  $\sigma_2 = S(U_2)$, 
$A=U_2\setminus U_1$, and $B=U_3\setminus U_2$.
We have 
\begin{align*}
\sigma_2 &= S(U_2) = S(U_1 \cup A) = \sigma_1\oplus S(A)\\
\sigma_1 &= S(U_3) = S(U_1 \cup A \cup B) = S(U_1\cup A\cup B \cup A) = S(U_3)\oplus S(A) = \sigma_1 \oplus S(A)  
\end{align*}
Therefore, $\sigma_1 = \sigma_2$.
\end{proof}

For a sketch $\sigma$, we denote by
$S^{-1}(\sigma) := \{U\in 2^{\mathcal{U}} \mid S(U)=\sigma\}$ the
set of all subsets with the same sketch. We show that $S^{-1}(\sigma)$ is closed under set union:
\begin{equation} \label{samesketchunion:claim}
    \forall U,V\in 2^{\mathcal{U}}, \,  U,V \in S^{-1}(\sigma) \implies U\cup V \in S^{-1}(\sigma)
\end{equation}
(this because $S(U\cup V) = S(U)\oplus S(V) = \sigma \oplus \sigma = \sigma$, using~\eqref{selfcompose:eq}.)

In particular, it follows from associativity of composition and \eqref{samesketchunion:claim} that for any sketch $\sigma$, there is a maximum set (unique maximal set) in $S^{-1}(\sigma)$: 
\begin{equation}
    \forall \sigma\in S(2^{\mathcal{U}}),\, \bigcup_{U\in S^{-1}(\sigma)} U  \in S^{-1}(\sigma)\ .
\end{equation}

\begin{definition} [Maximum Set of a Sketch]
    For a sketch $\sigma$, denote the maximum subset for the sketch by $\MaxSet(\sigma) := \bigcup_{U\in S^{-1}(\sigma)} U$. 
\end{definition}

\begin{definition} [Partial Order on Sketches]
For two sketches $\sigma_1$ and $\sigma_2$, we say that
$\sigma_1 \prec \sigma_2$  if $\mathtt{MaxSet}(\sigma_1) \subset \mathtt{MaxSet}(\sigma_2)$.
\end{definition}

\begin{claim} [Containment order implies sketch order] \label{containmentsketchorder:claim} 
\[
      U\subset V \implies S(U) \prec S(V)\ . 
      \]
\end{claim}
\begin{proof}
    \begin{align*}
        S(V) &= S(U \cup V) = S(U) \oplus S(V) = S(\MaxSet(U))\oplus S(\MaxSet(V)) \\ &= S(\MaxSet(U) \cup \MaxSet(V))\ .
    \end{align*}
   Therefore 
   $\MaxSet(U) \cup \MaxSet(V) \subset \MaxSet(V)$ and therefore
   $\MaxSet(U)\subset \MaxSet(V)$.
\end{proof}
Note that the other direction is not necessarily true, that is, it is possible that $S(U)\prec S(V)$ and $U\not\subset V$: Consider bottom-1 sketches, where the core of each sketch is a single key and the sketch order corresponds to the priority order over keys.

It directly follows from Claim~\ref{containmentsketchorder:claim} that if $U$ is a subset of $\MaxSet(\sigma)$ then $S(U)\prec \sigma$ and if $U$ contains  a core of $\sigma$ then $\sigma\prec S(U)$:
 \begin{align}
         U\subset \MaxSet(\sigma) &\implies S(U)\prec \sigma \label{subsketch:eq}\\
         C\in\Cores(\sigma) \text{ and } C\subset U &\implies \sigma \prec S(U)\ . \label{supsketch:eq}
 \end{align}

We can now characterize $S^{-1}(\sigma)$: A subset has a sketch $\sigma$ if only if it contains a core of $\sigma$ and is contained in $\MaxSet(\sigma)$:
\begin{claim} \label{sketchissigma:claim}
Let $\sigma$ be a sketch and $U\in 2^{\mathcal{U}}$ then
\[
S(U)=\sigma \, \iff \, \exists C\in\Cores(\sigma) \text{ such that }
C\subset U\subset \MaxSet(\sigma)\ .
\]
\end{claim}
\begin{proof}
$\implies$ direction: By definition of $\MaxSet(\sigma)$, if $S(U)=\sigma$ then $U\subset \MaxSet(\sigma)$. By definition of $\Cores(\sigma)$, if $S(U)=\sigma$ then $U$ must contain a core of $\sigma$.  $\impliedby$ direction:
  Applying Claim~\eqref{midpoint:claim}, since $C\subset U\subset \MaxSet(\sigma)$ and $C,\MaxSet(\sigma)\in S^{-1}(\sigma)$ it follows that $S(U) \in S^{-1}(\sigma)$. 
\end{proof}

As an immediate corollary of Claim~\ref{sketchissigma:claim} we obtain that any set must contain a core of its sketch:
\begin{equation} \label{coreinset:eq}
    \forall U,\  \exists C\in \Cores(S(U)) \text{ such that } C\subset U\ .
\end{equation}

\subsection{Rank Lemma}

We establish that if the sketch representation is limited to $k$ bits, then the rank can be at most $k$:
\begin{lemma}  [Core Size; Restatement and proof of Lemma~\ref{coresize:lemma}] \label{coresizeklimit:claim}
    If $|\sigma(U)|\leq k$ for all $U\in 2^{\mathcal{U}}$ then 
    $\rank((S,\mc U)\leq k$, that is,
    for all $\sigma\in \{S(U) \mid U\in 2^{\mathcal{U}}\}$, and all $C\in \Cores(\sigma)$, 
    $|C|\leq k$. 
\end{lemma}
\begin{proof}
    Let $C\in\Cores(\sigma)$ be a core of a sketch $\sigma$. We show that the sketches $S(U)$ for $U\in 2^C$ must be distinct from each other, that is,
    $|\{S(U) \mid U\in 2^C \}|=2^{|C|}$.
    Assume to the contrary that there are two subsets 
    $C_1,C_2\in 2^C$ such that $S(C_1)=S(C_2)$ but there is $y\in C_1\setminus C_2$.
    \begin{align*}
        \sigma &= S(C) = S(C_1 \cup C\setminus C_1) = S(C_1) \oplus S(C\setminus C_1) \\
         &= S(C_2) \oplus S(C\setminus C_1) = S(C_2 \cup C\setminus C_1)\ .
    \end{align*}
     Now consider $C' := C_2 \cup C\setminus C_1$. It holds that $C'\subset C$ and this is strict containment, that is $C' \subsetneq C$, because $y\in C$ and $y\not\in C'$. From the minimality of cores, we must have $S(C')\not= \sigma$ and we get a contradiction.
\end{proof}

\subsection{Termination Lemma}

We consider the sampling of $Q\sim \Bern[q]^{\mathcal{U}}$ as it unfolds by layers. 
Let $Q_i := Q\cap A_i$ be the sampled keys from layer $i$. Observe that $Q_i \sim \Bern[q]^{A_i}$ are independent random variables and hence their conditional distribution on fixing other $Q_j$'s remains the same.
\begin{lemma} [Termination Condition] \label{termination:lemma}
    To establish Lemma~\ref{composablepool:lemma} it suffices to show that under the respective conditions:
    For any $q\geq \qmin$, with probability at least $1-\delta$ (over sampling of $Q$),    
    the condition 
\begin{equation} \label{terminationcond:eq}
S(Q_{\leq i})= S\left(Q_{\leq i} \cup A_{i+1} \right)\ .
\end{equation}    
    is satisfied at some $i\leq \ell$.
\end{lemma}
\begin{proof}
 
We can stop considering the sampling at layer $i$ when the results of the sampling in higher layers do not matter, that is, for any mask $M\in 2^{\mathcal{U}}$ and any possible $Q_j$ for ($j>i$), the sketch $S(Q\cup M)$ is determined from $Q_1,\ldots, Q_i$. Formally, 
\begin{equation} \label{terminationcond2:eq}
\forall M\in 2^{\mathcal{U}},\ S(M\cup Q_{\leq i})= S\left(M\cup Q_{\leq i} \cup A_{>i} \right)\ .
\end{equation}
Since $M\cup Q_{\leq i} \subset Q\cup M \subset M\cup Q_{\leq i} \cup A_{>i}$, from Claim~\ref{midpoint:claim}, \eqref{terminationcond2:eq} implies that we determined $S(Q\cup M)$.
We next simplify this condition.

If $S(A)=S(B)$ then from composability for any $M$,
  $S(A\cup M)=S(A)\oplus S(M) = S(B)\oplus S(M) = S(B\cup M)$.
  Therefore for \eqref{terminationcond2:eq} to hold it suffices to consider $M=\emptyset$:
  \begin{equation} \label{emptyM:eq}
      S(Q_{\leq i})= S\left(Q_{\leq i} \cup A_{>i} \right)\ .
  \end{equation} 
  Next, recall from peeling that $A_{i+1} \in \Cores(S(A_{>i}))$ and hence $S(A_{i+1}) = S(A_{>i})$. Using composition rules
  \[
  S\left(Q_{\leq i} \cup A_{>i} \right) = S(Q_{\leq i})\oplus S(A_{>i}) = 
  S(Q_{\leq i})\oplus S(A_{i+1}) = S\left(Q_{\leq i} \cup A_{i+1} \right)\ .
  \]
   This concludes the proof by showing that \eqref{emptyM:eq} is equivalent to \eqref{terminationcond:eq}.
\end{proof}


%% file: 800-termination-proofs.tex
\section{Termination Proofs for Unique Cores} \label{uniqcores:sec}

  In this section we establish the termination condition \eqref{terminationcond:eq} for $\ell$ as stated in Lemma~\ref{composablepool:lemma} holds for the special case of sketching maps where sketches have unique cores:

\begin{definition} [Unique Cores] \label{MUcores:def}
    A sketching map $(S,\mathcal{U})$ has \emph{unique cores} if 
    for all $\sigma$, $|\Cores(\sigma)|=1$. In this case, we use the notation $\core(\sigma)$ for the core. A sketching map with unique cores is \emph{monotone} if  $U\subset V \implies |\core(S(U)|\leq |\core(S(V))|$.\footnote{This monotonicity requirement may appear to be weaker than in Definition~\ref{monotonecomposablemap:def} but our proof will reveal it is equivalent. Regardless, the result is only stronger if established under weaker assumptions.}
\end{definition}

Section~\ref{propuniqcores:sec} includes preliminaries for unique cores. The termination proof is concluded  in Section~\ref{uniqcorestermination:sec} for general unique cores and 
in Section~\ref{monotoneuniqcore:sec} for monotone unique cores. In Section~\ref{tightnonmonotone:sec} we give examples that show that our quadratic bound on the size of the pool for maps with unique cores is tight for peeling pool constructions.


\subsection{Properties of Unique Cores} \label{propuniqcores:sec}
We establish some properties that hold with unique cores.
We sometimes abuse notation and write $\core(U) := \core(S(U))$.

If a key in $\MaxSet(\sigma_2)$ is in the core of a ``higher'' sketch $\sigma_1 \succ \sigma_2$, then it must also be ``inherited'' in  $\core(\sigma_2)$.
\begin{claim} [Inheritance for Unique Cores] \label{uniqcoreinheritance:claim}
    If $\sigma_1 \succ \sigma_2$ then
    $\core(\sigma_1) \cap \MaxSet(\sigma_2) \subset \core(\sigma_2)$.
\end{claim}
\begin{proof}
Let $A = \core(\sigma_1) \setminus \MaxSet(\sigma_2)$.
It follows that
\[
\core(\sigma_1) \subset A \cup \MaxSet(\sigma_2) \subset \MaxSet(\sigma_1)
\]
(since $A\subset \core(\sigma_1)$ and by assumption $\MaxSet(\sigma_2) \subset \MaxSet(\sigma_1)$).

Therefore, from Claim~\ref{midpoint:claim}, 
$S(A \cup \MaxSet(\sigma_2)) = \sigma_1$.
Therefore, 
\begin{equation} \label{pp1:eq}
    S(A)\oplus S(\MaxSet(\sigma_2))=S(A)\oplus \sigma_2 = \sigma_1\ .
\end{equation}
Now assume to the contrary that the statement does not hold.
Then $\core(\sigma_1)\subsetneq A \cup \core(\sigma_2)$ and hence from Claim~\ref{sketchissigma:claim}
$S(A \cup \core(\sigma_2))\neq \sigma_1$. Therefore
$S(A) \oplus S(\core(\sigma_2)) = S(A)\oplus \sigma_2 \not= \sigma_1$. We get a contradiction to \eqref{pp1:eq}.
\end{proof}

We show that the core of the sketch of the union of subsets must be a subset of the union of their cores:
\begin{claim} [Core of Union is in Union of Cores] \label{mergeuniq:claim}
    \begin{equation*} 
    \forall U_1,U_2 \in 2^{\mathcal{U}},\ \core(S(U_1\cup U_2)) \subset \core(S(U_1)) \cup \core(S(U_2))
\end{equation*}
\end{claim}
\begin{proof}
\begin{align*}
    S(U_1\cup U_2) &= S(U_1)\oplus S(U_2) = S(\core(U_1)) \oplus S(\core(U_2)) \\ &= S(\core(U_1)\cup \core(U_2))\ .
\end{align*}
Hence, using~\eqref{coreinset:eq}
\[ \core(S(U_1\cup U_2)) = \core(S(\core(U_1)\cup \core(U_2)))\subset \core(U_1)\cup \core(U_2)\ .\]
\end{proof}

\subsection{Termination Proof for Unique Cores} \label{uniqcorestermination:sec}
We show that
condition~\eqref{terminationcond:eq} (the termination condition) is satisfied with probability at least $(1-\delta)$ after at most $\ell = (k+O(\sqrt{k} \ln(1/\delta))/\qmin$ layers.

Observe that with unique cores, the peeling process simplifies to
$A_{i+1} \gets \core(S(\mathcal{U})\setminus \bigcup_{j\leq i} A_j)$ and at each layer there is a new sketch.

We specify members of $\core(S(Q))$ layer by layer, after  fixing the sampling $Q_i$.  We use the following claim:
  \begin{claim} \label{pu2:claim}
If the termination condition \eqref{terminationcond:eq} is not satisfied at the beginning of step $i$ (after fixing $(Q_j)_{j< i}$), then there is at least one key in $A_i$ that is \emph{active}, that is, if it is sampled into $Q_i$, it must be included in $C_i:= \core(S(Q))\cap A_i$.
  \end{claim}

It follows from Claim~\ref{pu2:claim} that at each layer before termination, there is probability at least $q$ of determining at least one additional member of $\core(S(Q))$.
Since from Claim~\ref{coresizeklimit:claim}, a core can be of size at most $k$, then in particular, $|\core(Q)|\leq k$, and we must terminate when or before a $k$th key of $\core(S(Q))$ is determined. 
    
We apply Chernoff bounds. The probability that there are fewer than $k$ successes in $\ell$ $\Bern[q]$ i.i.d.\ Bernoulli trials is at most
$e^{-(q\ell -k)^2/(2q\ell)}$. We can verify that with
$\ell = (k+O(\sqrt{k} \ln(1/\delta))/q$ we obtain that  $e^{-(q\ell -k)^2/(2q\ell)}\leq \delta$. This concludes the proof.

It remains to prove Claim~\ref{pu2:claim}:
\begin{proof}[Proof of Claim~\ref{pu2:claim}]
Consider the sequence
\begin{align*}
   U_1 &=A_{\geq 1} = \mathcal{U}& \\
   U_i &= Q_{<i} \cup A_{\geq i} & \text{; for $i>1$}
\end{align*}
Observe that this is a 
non-increasing sequence 
$U_1\supset U_2 \cdots \supset U_\ell \cdots \supset Q$ of supersets of $Q$ and that
the set $U_i$ is determined once we fixed the sampling $Q_j$ for $j<i$. From Claim~\ref{containmentsketchorder:claim}, 
$S(U_1)\succ\cdots \succ S(U_\ell)\cdots \succ S(Q)$.

Since, $S(U_i)\succ S(Q)$, from Claim~\ref{uniqcoreinheritance:claim}, all keys in $Q$ that are in $\core(S(U_i))$ must be in $\core(S(Q))$. That is,
$\core(S(U_i))\cap Q \subset \core(S(Q))$.
Similarly, all keys in $Q$ that are in $\core(S(U_i))$ must be in $\core(S(U_j))$ for $j>i$.

Let $\mathtt{Active}_i = \core(S(U_i))\setminus Q_{<i}$ be the set of keys in $\core(S(U_i))$ that are not included in
$Q_j$ for $j<i$. It follows that if any key of $\mathtt{Active}_i$ is sampled into $Q$, it must be in $\core(S(Q))$.

We show that $\mathtt{Active}_i \subset A_i$:
\begin{claim} \label{coresui:claim}
$\core(S(U_i)) \subset A_i \cup Q_{<i}$
\end{claim}
\begin{proof}
In the first step, $A_1 = \core(S(U_1))$ and therefore 
$\mathtt{Active}_1 =A_1$ and any key that is sampled will be in $S(Q)$. 
$Q_1 \subset A_1$ and  
$C_1 = Q_1 \subset \core(S(Q))$.

For $i>1$, $U_i = Q_{<i} \cup A_{\geq i}$ and hence 
\begin{align*}
    \core(S(U_i)) &\subset \core(S(Q_{<i})) \cup \core(A_{\geq i}) &\, ;\text{Claim~\ref{mergeuniq:claim}}\\
    &\subset Q_{<i} \cup \core(A_{\geq i}) &\, ;\text{Eq. \eqref{coreinset:eq}}\\
    & \subset Q_{<i} \cup A_i &\, ;\text{ Peeling def}
\end{align*}
\end{proof}

We established that if there are active keys, 
that is $|\mathtt{Active}_i|\geq 1$, then they must be included in $A_i$. We obtain that $C_i = \mathtt{Active}_i \cap Q_i$.
We now show that if there are no active keys, then the termination condition is met.

For $i=2$, observe that if $Q_1=A_1$ (all keys in $A_1$ are sampled), then
the termination condition \eqref{terminationcond:eq} is satisfied at $i=2$:
It holds that $S(Q_1)=S(A_1)=S(A_{i\geq 1})$ (since by the peeling $A_1=\core(S(A_{\geq 1}))$).
Since $A_1\subset A_1\cup A_2\subset  A_{\geq 1}$, it follows from Claim~\ref{midpoint:claim} that $S(A_1\cup A_2)=S(A_1)$.
Combining, we obtain that $S(Q_1)=S(Q_1\cup A_2)$.

If $|\mathtt{Active}_i|=0$, then 
$\core(S(U_i)) \subset Q_{<i}$. In this case, $S(U_i)=S(Q_{<i})$ and
condition \eqref{terminationcond:eq} holds at step $i-1$.
This concludes the proof of Claim~\ref{pu2:claim}.
\end{proof}


\subsection{Termination Proof for Monotone Unique Cores} \label{monotoneuniqcore:sec}
We now consider the case of \emph{monotone} unique cores (see Definition~\ref{MUcores:def}).

\begin{claim}
    With probability at least $(1-\delta)$, the termination condition is satisfied within
$\ell = O(\log(1/\delta)/\qmin)$ steps.
\end{claim}
\begin{proof}
We piggyback on our proof for unique cores, defining $U_i$, $C_i$, $Q_i$, and $\mathtt{Active}_i\subset A_i$ the same way.

It follows from monotonicity that $(|\core(S(U_i))|)_{i\geq 1}$  is non-increasing and that for all $i$, $|\core(S(U_i))| \geq |\core(S(Q))|$.  
From Claim~\ref{coresui:claim}, it follows that
\[
\core(S(U_i)) = C_{<i} \cup \mathtt{Active}(A_i)\ .
\]
Since these sets are disjoint ($C_j \subset A_j$), it follows that
\begin{equation} \label{active:eq}
|\core(S(U_i))| = \sum_{1\leq j<i} |C_j| + |\mathtt{Active}(A_i)|\ .
\end{equation}

Since $|\core(S(U_i))|$ are non-increasing and $\sum_{j<i} C_i$ is non-decreasing with $i$, it follows from \eqref{active:eq} that  
$|\mathtt{Active}_i|$ is non-increasing. Moreover, $|\mathtt{Active}_i|$ decreases in expectation by a factor of $\cdot(1-q)$. 

Recall that when there are no active keys, that is 
$|\mathtt{Active}_i|=0$, the termination condition is satisfied. At layer $i=1$, $\mathtt{Active}_1=A_1$ and hence
$|\mathtt{Active}_1|\leq |A_1|\leq k$.

Solving for each of the $k$ ``active positions,'' the probability that the process does not terminate after $\ell$ steps is bounded by the probability that $\Pr[\Geom[q] > \ell]$. Taking a union bound over $k$ ``positions'' it suffices to choose $\ell$ so that this probability is at most $\delta/k$. Solving for $(1-q)^\ell \leq \delta/k$ we obtain that $\ell \approx \ln(k/\delta)/q = O(\ln(k/\delta)/\qmin)$ suffices for eliminating all ``active positions'' with probability at least $1-\delta$. This concludes the proof.  
\end{proof}

\subsection{Tightness of Analysis} \label{tightnonmonotone:sec}

Our $\tilde{O}(k)$ pool size result for monotone unique cores is asymptotically tight within logarithmic factors. We show that the
$\tilde{O}(k^2)$ pool size result is tight for  
composable sketching maps with unique cores. We show examples of (non monotone) sketching maps for which $\tilde{\Omega}(k)$ layers of size $\tilde{\Omega}(k)$ each are necessary for the pool, and hence the pool size of $\tilde{\Omega}(k^2)$ is necessary when using peeling pool constructions. 

\begin{example} \label{quadraticpool:example}
Our sketching map $(S,\mathcal{U})$ is constructed from two composable maps $(S_i,\mathcal{U}_i)$ for $i=1,2$ with disjoint ground sets and core sizes at most $k'=k/2$. We use $\mathcal{U} := \mathcal{U}_1 \uplus \mathcal{U}_2$ and $S(U)=S(U_1\uplus S_2) = (S_1(U_1),S_2(U_2))$.
Observe that such a combination is a composable map with cores size at most $k$. For a subset $U$ where 
$U_i = U\cap \mathcal{U}_i$, it holds that
$\core(U) = \core_{S_1}(U_1) \cup \core_{S_2}(U_2)$. A core peeling of $S$ is a layer-wise union of core peelings of $S_1$ and $S_2$.
We choose $(S_2,\mathcal{U}_2)$ to be a $k'$-mins sketch. 
Recall that all layers $(A^2_j)$ in the core peeling of $S_2$ have size $k'$.
It remains to describe $(S_1,\mathcal{U}_1)$, which is designed to have a ``long and thin'' core peeling.
The map $S_1$ is defined with respect to disjoint sets $(R_i)_{i\in[k'/W]}$ of size $|R_i|=W=\log(10k/\delta)$. The sketch is designed so that for a set $U$, we can determine from $S(U)$ the smallest index $i$ for which $U$ contains $R_i$: $\arg\min_{i\in [k'/W]} R_i\subset U$. 
The composable sketch $S_1(U)$ must encode the smallest such $i$ and for all $j<i$, also encode $U\cap R_i$. In particular, 
$\core(U) = R_i\cup \biguplus_{j<i} U\cap R_j$.
A core peeling of $S_1$ is simply $A^1_i \equiv R_i$. 
Note that with $q\leq 1/2$, the probability that all keys in a particular $R_i$ are sampled is at most $\delta/(10 k)$. Hence the probability of termination before layer $k'$ is $< \delta$ and a determining pool for $S_1$ must use all $k'$ layers.
This means that a determining pool for the sketching map $S$ must also use $k$ layers and in total these layers include more that $(k')^2=k^2/4$ keys.
\end{example}

\ignore{
  The Example~\ref{quadraticpool:example} can be mitigated by refining our peeling construction so that we delete from each $A_i$ the keys that have probability $< \delta/k^2$ of surviving the sampling of layers $A_j$ ($j<i$), that is, remaining at the core of $Q_{<i}\cup A_i$. With this modification, the keys in the peeling of $S_2$ beyond layer $\log(k/\delta)$ are deleted and we remain with a pool of size $\tilde{O}(k)$. We next describe an example sketching map where this strategy is also insufficient. 
  
\eccomment{to clean up}

\begin{example}
Each $A_i$ is of size $i\cdot W$ where $W=\log(10k/\delta)$. The compositions says that if fewer than $0.25 W$ of the keys of $A_i$ are included then we need specific $0.75 W$ of the keys of the $A_{i+1}$ but the specific is distributed equally (hash of what was selected in prior layer). This "introduces" new capacity at each layer but keeps all of $A_i$ in the game (probability of surviving distributed equally and $>1/k$ so can not be trimmed from peeling. So we can have $k/J$ layers and still able to store the sketches (record each subset of relevant keys is sampled from each layer). 
\end{example}
 }

 \section{Termination Proof for the General Case} \label{proofgen:sec}

In this section we establish that the termination condition (see Lemma~\ref{termination:lemma}) holds for general composable sketches with $\ell=O(k\log(k/\delta))$ and for monotone composable sketches with $\ell=O(\log(k/\delta))$. This would conclude the proof of Theorem~\ref{composablepool:theorem}.

\subsection{Termination Proof for Monotone Composable Sketches}

We reuse the notation introduced in Section~\ref{uniqcores:sec}. The following definition will be handy:
\begin{definition} [Residual Set]
   For a set of subsets $P\subset 2^{\mathcal{U}}$ and $A\in 2^{\mathcal{U}}$, the \emph{residual} set of $P$ with respect to $A$, $P\mid_A$, is the set of all \emph{minimal} subsets in
$P\setminus A := \{U\setminus A \mid U\in P\}$. 
\end{definition}

The monotonicity condition in \cref{monotonecomposablemap:def} can equivalently be stated as follows.
For all $U_1\subset U_2$ and $A$,
The cardinality of sets in
$\InCores(U_1 \cup A) \mid_A$ is at most that of sets in
$\InCores(U_2 \cup A) \mid_A$.

\begin{claim} \label{samesize:claim}
    For a monotone composable $(S,\mathcal{U})$, all cores $\Cores(\sigma)$ of the same sketch $\sigma$ have the same cardinality and for all $A$, all residual cores
$\Cores(\sigma)\mid_A$ have the same cardinality. 
\end{claim}
\begin{proof}
    If two cores $C_1,C_2$ are such that $|C_1|<|C_2|$ then 
    consider $U_1=C_2$ and $U_2= C_1\cup C_2$ and note that $U_1\subset U_2$. Then $C_1\in \InCores(C_1\cup C_2)$ and $C_2\in \InCores(C_2)$ which contradicts monotonicity.
\end{proof}


 For each $i$ and fixing $Q_{<i}$, it follows from Claim~\ref{samesize:claim}, all residual cores 
 $C\in \InCores(A_{\geq i}\cup Q_{< i})\mid Q_{< i}$
 have the same cardinality. We denote this value by $r_i$. 

\begin{claim}
    If $r_i=0$ then the termination condition~\eqref{terminationcond:eq} is satisfied for $i-1$.
\end{claim}
\begin{proof}
    If $r_i=0$ then there is $C\in \InCores(A_{\geq i}\cup Q_{< i})$ that is $C\subset Q_{<i}$.
Note that $C\subset Q_{<i} \subset A_{\geq i}\cup Q_{< i}$. From Claim~\ref{midpoint:claim}, since $S(C)=S(A_{\geq i}\cup Q_{< i})$ then
$S(Q_{< i})=S(A_{\geq i}\cup Q_{< i})$.
\end{proof}

Since $A_i$ is a core of $A_{\geq i}$, from Claim~\ref{coreofunion:claim} there must be a residual core
$A'_i\in \InCores(A_{\geq i}\cup Q_{< i})\mid Q_{< i}$ such that
$A'_i\subset A_i$. Recall that $|A'_i|=r_i$.

\begin{claim}
    Fixing $Q_i$, it holds that $r_{i+1} \leq r_i - |A'_i\cap Q_i|$
\end{claim}
\begin{proof}
Since
$A'_i\in \InCores(A_{\geq i}\cup Q_{< i})\mid Q_{< i}$, there must be
$A''_i\subset A'_i$ such that
\[
A''_i\setminus Q_i \in \InCores(A_{\geq i}\cup Q_{< i})\mid_{Q_{< i}\cup Q_i} =
\InCores(A_{\geq i}\cup Q_{< i})\mid_{Q_{< i+1}}\ .
\]
From the monotonicity property applied to the sets
$A_{\geq i+1}\cup Q_{< i+1}\subset A_{\geq i}\cup Q_{< i+1}$ and $Q_{< i+1}$, it follows that
$r_{i+1}\leq |A''_i\setminus Q_i|$. The claim follows  using that $|A''_i\setminus Q_i| \leq |A'_i\setminus Q_i|= r_i - |A'_i\cap Q_i|$.
\end{proof}

We are now ready to conclude the termination proof. At each step, when sampling $Q_i$, $r_i$ decreases in expectation by a factor of $(1-q)$.  The argument for bounding the number of steps until $r_i=0$ with probability at least $(1-\delta)$ is then exactly the same as for monotone unique cores in Section~\ref{monotoneuniqcore:sec}.

\subsection{Termination Proof for Composable Maps}
We reuse the notation introduced in Section~\ref{uniqcores:sec}:
Let
\begin{align*}
   U_1 &=A_{\geq 1} & \\
   U_i &= Q_{<i} \cup A_{\geq i} & \text{; for $i>1$}
\end{align*}
Recall that this is a 
non-increasing sequence 
$U_1\supset U_2 \cdots \supset U_\ell \cdots \supset Q$ of supersets of $Q$ and that
the set $U_i$ is determined once we fixed the sampling $Q_j$ for $j<i$. 
\ignore{
Let $P = \Cores(S(U_\ell)) \mid_{Q_{<\ell}} $.
\begin{claim}
  $P=\{\emptyset\}$ if and only if the termination condition~\eqref{terminationcond:eq} is met at some $i\leq \ell$.  
\end{claim}
\begin{proof}
From Claim~\ref{emptyresidual:claim}, it suffices to show that the termination condition is met at some $i\leq \ell$ if and only if $Q_{<\ell}$ contains a core of $S(U_\ell)$.
If $Q_{<\ell}$ contains a core $C\in \Cores(S(U_\ell))$ then from Claim~\ref{midpoint:claim}, since we also have $C\subset Q_{<\ell}\subset U_\ell$ it holds that $S(Q_{<\ell}) = S(U_\ell)=S(C)$ and hence \eqref{terminationcond:eq} is met at $\ell$.
If \eqref{terminationcond:eq} is met at $i\leq \ell$ then 
$S(U_i) = S(Q_{<i})$. Since $Q_{<i} \subset Q_{<\ell} \subset U_\ell \subset U_i$ it holds from Claim~\ref{midpoint:claim} that 
$S(Q_{<\ell)} = S(U_\ell) = S(Q_{<i})$.
\end{proof}

It remains to consider the case that $P\neq \{\emptyset\}$, which corresponds to the case that the termination condition was not met at $\ell$. We need to show that the probability of that is at most $\delta$.

Let 
$P_i$ ($i\in [\ell]$) be the residual cores of $S(U_\ell)$ with respect to $Q_{< \ell, -i} := \bigcup_{j\in [\ell], j\not= i} Q_j$.
\begin{equation}
    P_i := \Cores(S(U_{\ell}))\mid_{Q_{< \ell, -i}}
\end{equation}
Since $P\neq\emptyset$ and for all $i$, $Q_{< \ell, -i}\subset Q_{<\ell}$ it follows from Claim~\ref{monotonicityres:claim} that $P_i\neq\{\emptyset\}$.

} The proof proceeds as follows:
\begin{enumerate}
    \item 
    We establish in Lemma~\ref{yield:lemma} that conditioned on no termination,
    each step $i$ has probability at least $q$ of having a nonempty set of ``yield'' keys (to be defined in the statement of the lemma) $C_i\subset Q_i$. This for any fixed sampling $Q_j$ at other steps $j\not=i$.
Hence, by reusing the same standard Chernoff argument as in Section~\ref{uniqcorestermination:sec}, we obtain that with probability $(1-\delta)$, there are at least $k$ steps with a nonempty yield when we perform $\ell$ steps. Equivalently, the probability that there were $k$ or fewer yields in $\ell$ steps is at most $\delta$.
    \item
    We show that before termination we can have at most $k$ nonempty yields in total. Note that since $C_i\subset A_i$, the yields from different steps are disjoint. The total yield is $C:= \bigcup_{i<\ell} C_i$ and therefore it suffices to argue that 
    $|C| = \sum_{i<\ell} |C_i|\leq k$. This is established
    in Lemma~\ref{kyields:lemma} that states that $C$ must be a core of $S(C)$. From Claim~\ref{coresize:lemma}, a core can be of size at most $k$, and therefore it must hold that $|C|\leq k$.
\end{enumerate}
Combining the two, we obtain that the 
probability that the termination condition was not satisfied in $\ell$ steps is at most the probability of having $k$ or fewer yields in $\ell$ steps, which is at most $\delta$.

It remains to prove the two lemmas.
The following generalizes Claim~\ref{mergeuniq:claim}:
\begin{claim} [Cores of a Union in Union of Cores] \label{coreofunion:claim}
Let $U_1,U_2 \in 2^{\mathcal{U}}$  be disjoint.
For any $C_1\in \InCores(U_1)$ and $C_2\in \InCores(U_2)$ 
there is $C\in \Cores(S(U_1\cup U_2))$ 
such that $C\subset C_1\cup C_2$.

Conversely, for any $C\in \InCores(U_1\cup U_2)$ there is 
$C_1\in \InCores(U_1)$ and $C_2\in \InCores(U_2)$ such that
$C\subset C_1\cup C_2$.
\end{claim}
\begin{proof}
\begin{align*}
        S(U_1\cup U_2)= S(U_1)\oplus S(U_2) = S(C_1) \oplus S(C_2) = S(C_1\cup C_2)
\end{align*}
From \eqref{coreinset:eq} 
there must be a core of $S(C_1\cup C_2)$ (which is the same as a core of $S(U_1\cup U_2)$) that is a subset of $C_1\cup C_2$. This concludes the proof of the first part of the statement.

Let $C'_i = U_i\cap C$.
\begin{align*}
    S(C) &= S(C'_1 \cup C'_2) = S(C'_1) \oplus S(C'_2)
\end{align*}
From \eqref{coreinset:eq}, there is $C_i\subset C'_i$ such that $C_i\in \Cores(S(C'_i))$.
Therefore 
\begin{align*}
    S(C'_1) \oplus S(C'_2) = S(C_1) \oplus S(C_2) = S(C_1\cup C_2)\ .
\end{align*}
This concludes the proof of the second part of the statement.
\end{proof}

\begin{claim} [Core Inheritance] \label{inheritancegen:claim}
    Let $A,B \in 2^{\mathcal{U}}$ be such that
    for all $C\in \InCores(A\cup B)$, it holds that $A\subset C$.
    Then for any $A'\subset A$ and $B'\subset B$ it holds that for
    all $C\in \InCores(A'\cup B')$, $A'\subset C$.
\end{claim}
\begin{proof}
    \begin{align*}
        \InCores(A\cup B) = \InCores( (A'\cup B') \cup (A\setminus A' \cup B\setminus B'))\ .
    \end{align*}
    Applying Claim~\ref{coreofunion:claim} it follows that for any
    $D' \in \InCores(A'\cup B')$ and $D''  \in \InCores((A\setminus A') \cup (B\setminus B'))$ there is a core $D\in \InCores(A\cup B)$ such that $D\subset D'\cup D''$. Since $A\subset D$ and the two parts are disjoint, it must hold that $A'\subset D'$ (and $A\setminus A'\subset D''$). Therefore, any core $D' \in \InCores( (A'\cup B')$ must contain $A'$.
\end{proof}


\begin{lemma} [Probability of yield] \label{yield:lemma}
Consider the sampling of $Q_j$.
Fix a step $i$ and fix $Q_{<\ell,-i}$ (that is, $Q_j$ for all $j\in[\ell-1],j\not= i$. 
Assume that termination does not hold based on $Q_{<\ell,-i}$ alone, that is, 
\begin{equation}\label{terminationminusi:eq}
    S(Q_{<\ell,-i}) \neq S(Q_{<\ell,-i} \cup A_{\geq\ell})\ .
\end{equation}

Then with probability at least $q$ over the sampling of $Q_i$, there is a non-empty $C_i\subset Q_i$ such that if sampled, must be contained in all $\InCores(Q_{< \ell, -i} \cup C_i)$:
    \[
    \forall D\in \InCores(Q_{< \ell, -i} \cup C_i),\ C_i \subset D\ .
    \]
We refer to the set $C_i$ as the \emph{yield} of step $i$.
\end{lemma}
\begin{proof}

Consider  
\[
V_i := Q_{<\ell, -i} \cup A_{\geq \ell} \cup A_i\ .
\]

Consider the set $\mathcal{D}$ of all $D\in \InCores(V_i)$ such that $D\subset Q_{<\ell, -i} \cup A_i$.
\begin{claim}
    $\mathcal{D}$ is not empty and for all $D\in\mathcal{D}$, $D\cap A_i$ is not empty.
\end{claim}
\begin{proof}
We show that this set is not empty.
$A_i$ is a core of $S(A_{\ge i})$ and hence a core of
$S(A_{\geq \ell} \cup A_i)$.
Applying Claim~\ref{coreofunion:claim}, taking any $B\in\InCores(Q_{<\ell, -i})$
there exists $D\in \InCores(V_i)$ such that $D\subset B\cup A_i$. 

We also show any $D\in\mathcal{D}$ must have a nonempty $D\cap A_i$.
If this was not the case, then 
\[
D\subset Q_{<\ell, -i} \subset Q_{<\ell} \subset U_\ell = Q_{<\ell} \cup A_{\geq \ell} \subset Q_{<\ell, -i} \cup A_{\geq \ell} \cup A_i=V_i\ .
\]
Since $S(D)=S(V_i)$, applying \cref{midpoint:claim}, we obtain that $S(Q_{<\ell,-i})=S(U_\ell)$ \msmargincomment{the assumption of this lemma was that the termination condition does not hold at step $i$, but this shows that the termination condition holds at step $\ell$ (or maybe $\ell-1$), right?}\ecmargincomment{should be fixed} and therefore the termination condition is satisfied, which contradicts our assumption \cref{terminationminusi:eq}.
\end{proof} 

Let $A'_i\subset A_i$ be a minimal set in $\{D\cap A_i\mid D\in \mathcal{D}\}$. From the claim we know that there is a minimal nonempty such set.
 From minimality,
$A'_i$ must be contained in all cores in 
$\InCores(Q_{<\ell, -i} \cup A'_i)$. This is because
$S(Q_{<\ell, -i} \cup A'_i)=S(V_i)$ and all cores of $V_i$ that are contained in $Q_{<\ell, -i} \cup A'_i$ must include $A'_i$.

Let $C_i = Q_i\cap A'_i$. Note that since $|A'_i|>0$, with probability $\geq q$, $|C_i|\geq 1$.

It follows from Claim~\ref{inheritancegen:claim} that $C_i\subset A'_i$
must be contained in all cores in
$\InCores(Q_{<\ell, -i} \cup C_i)$. This concludes the proof of the lemma.
   \end{proof}

\begin{lemma} [There can be at most $k$ yields] \label{kyields:lemma}
 Assume that the termination condition~\ref{termination:lemma} does not hold at step $\ell$ and let $C_i$ for $i<\ell$ be the yield (possibly empty) at step $i$. Let $C:= \bigcup_{i<\ell} C_i$.
 Then $C\in \InCores(C)$.
\end{lemma}
\begin{proof}
Applying Claim~\ref{inheritancegen:claim} with Lemma~\ref{yield:lemma} and using that $\bigcup_{j<\ell, j\not=i} C_j \subset Q_{<\ell,-i}$ we obtain that for all $i<\ell$,
    \[
    \forall D\in \InCores\left(\bigcup_{j<\ell, j\not=i} C_j  \cup C_i\right),\ C_i \subset D\ .
    \]
    Rewriting we obtain that 
    \begin{align*}
        \forall i<\ell, \forall D\in \InCores\left(\bigcup_{j<\ell} C_j \right),\ C_i \subset D\ .
    \end{align*}
    Therefore using $C:= \bigcup_{j<\ell} C_j$ this is equivalent to the claim in the statement of the Lemma:
    \begin{align*}
        \forall D\in \InCores(C ),\ C \subset D\ .
    \end{align*}
    This concludes the proof.
\end{proof}

\ignore{
\begin{lemma} [Distinct Equivalence Classes] \label{distinctclass:lemma}
   For subsets $C_1\neq C_2 \in 2^C$ it holds that $[C_1]_{S(U_\ell)} \neq [C_2]_{S(U_\ell)}$.  
\end{lemma}
\begin{proof}
    \eccomment{formalize}
 This because each $c_i$ can not be removed from residual cores by the other $c_j$'s, since they are a subset of $Q_{< \ell, -i}$. The only way to remove each $c_j$ from residual cores of $S(U_\ell)$ is to explicitly include it.  

 We then need the claim that for $\sigma$ and $B\subset C$ and $c$, if it holds that $[C\cup\{c\}]_\sigma \neq [C]_\sigma$ then it holds that
$[B\cup\{c\}]_\sigma \neq [B]_\sigma$.

\end{proof}
}

%% file: 400-analysis.tex
\section{Analysis of the Universal Attack} \label{analysisattack:sec}

\newcommand{\ti}{^{(t)}}
\newcommand{\ts}{^{(r)}}

We establish \cref{metacorrectness:thm} and \cref{limitedpoolnaturalQR:lemma}.
We first make some simplifications that aid our analysis with no loss of generality. We will throughout use the assumption that $n$ is sufficiently large (i.e., larger than some absolute constant). Throughout this section, we will often show that something is true with probability $1-O(1/n)$, and then proceed to assume that that was actually the case; it is implicit that in the other case, we simply let the attack fail.

Finally, for the regime where $b/a \to 1$ (while $a, b$ are bounded away from $0, 1$) and $|L|/n \to 0$, we will assume that $b < (1+\e)a$ and $|L|/n < \e$. Indeed, we may assume that $b = (1+\e)a$, since increasing the threshold $B$ only makes the problem easier for the sketch. For the purpose of all asymptotic notation, we will still treat $\e$ as a constant unless it is explicitly stated that $\e \to 0$.

We will assume that at all steps the QR knows the current value of the mask $M$ (this only makes the QR more powerful). Denote the remaining elements of the determining pool by $L' = L \setminus M$, and let $\ell = |L'|$.

Now, suppose we run the interaction using \cref{metaattack:algo} for $r$ rounds, but for one of the steps, instead of giving the QR the sketch $S(U \cup M, X)$, we instead just give the QR the set $U \cap L' = U \cap L \setminus M$. If there exists a determining pool, the QR could then sample a sketch $S \sim S_{U \cap L'}$ instead of receiving it from the system.
Then, by \cref{coupling:remark}, we can couple the original interaction and this modified interaction so that the set $U$ and the sketch $S$ are the same, except with probability $\delta$\footnote{This argument assumes that in fact $M \subseteq L$, since that is a condition in the definition of a determining pool. We will show later that with probability $1-O(\inv n)$, after each step, $M \subseteq L$; if this ever does not hold, we allow the algorithm to fail.}. 
(Note that under the hypothesis of \cref{metacorrectness:thm} (and \cref{limitedpoolnaturalQR:lemma}), we have $\delta < n^{-5}$.) 

We can now repeat this modification for all $r$ steps of the algorithm. By a union bound, we can couple the original and modified interactions such that the queries and responses are all the same except with probability $r\delta = o(1)$. Thus, we have the following fact.

A similar argument holds under the conditions of \cref{limitedpoolnaturalQR:lemma} when QR just receives $W$ (or even a sample from $\mc D'(M,W)$ per \cref{limitedpool:def}) instead of $S(U \cup M, X)$.

\begin{observation} \label{altsketch:obs}
We may assume that the QR, instead of receiving $S(U \cup M, X)$ from the system, just receives $W \coloneq U \cap L'$.
\end{observation}
Thus, we will henceforth indeed assume that the QR only receives $W$. 

Finally, we define transparent keys to be keys that cannot be included in the sketch.
\begin{definition} [Transparent keys] \label{transparent:def}
    We say that keys in $\mathcal{U}\setminus L$ are \emph{transparent}. 
\end{definition}
Transparent keys have the property that their presence in the query is never in $W$ (no matter what the current mask $M$ is), and thus invisible to the QR. 

\subsection{Analysis of a single step}
In this subsection we focus on a single step of Algorithm~\ref{metaattack:algo}. Per \cref{altsketch:obs}, we assume the sketch distribution is 
$W \coloneqq U\cap L'$ for $U\sim \mathcal{Q}$. This is equivalent to sampling  $q$ with density $\nu(q)$ and then sampling a set from $\Bern[q]^{L'}$. 

We define the \textit{estimator map} $\phi : 2^{L'} \to [0, 1]$ as follows: for $W \subseteq L'$, $\phi(W)$ is the probability that the QR returns 1 if it receives $W$. Note that the function $\phi$ is independent of the randomness of the attacker in the current step.

For any key $x$, define 
\begin{equation} \label{scoreprob:eq}
p_\phi(x) = \E_{U\sim \mathcal{Q}}\left[\phi(U\cap L') \cdot \mathbf{1}_{x\in U} \right].
\end{equation}
This is the probability, for estimator map $\phi$, that the score $C[x]$ is incremented in the current step of \cref{metaattack:algo}.

\begin{claim} 
The score probability $p_\phi(x)$ is the same for all transparent keys $x\in \mathcal{U}\setminus L$ (Definition~\ref{transparent:def}).
\end{claim}    
\begin{proof}
Per \cref{altsketch:obs},
all actions of the QR including the selection of the 
estimator maps $\phi$ depend only on the intersection of the query with the pool and thus do not depend on which transparent keys are included in the query.
\end{proof}

We denote the score probability of transparent keys by $p^*_\phi$.
We establish that keys in $L'$ have \emph{on average}  higher probability of being scored than transparent keys.

Finally, let $\eta_\phi$ be the probability that the QR fails, given estimator map $\phi$. Formally, this is defined as
\begin{equation} \label{eta-def:eq}
\eta_\phi = \E_{U \sim Q}[\phi(U \cap L') \cdot \mathbf{1}_{|U \cup M| < A} + (1 - \phi(U \cap L')) \cdot \mathbf{1}_{|U \cup M| > B}].
\end{equation}

\begin{lemma} [Score advantage] \label{scoreadvantage:lemma}
There exist constants $\ell_0$ and $c_1, c_2$ such that, if $|L'| \ge \ell_0$, then
\begin{equation} \label{score-advantage:eq}
\ov p_\phi(L') \coloneqq \E_{x\sim \mathtt{Uniform}[L']} \left[ p_\phi(x)\right]  \ge  p^*_\phi + \f{c_1 - c_2 \eta_\phi}{|L'|}.
\end{equation}
Moreover, as $\e \to 0$, $c_1/c_2 \to 1/4$.
\end{lemma}
The proof of the score advantage lemma is technical and is deferred to Section~\ref{scoreadvantageproof:sec}.  The very high level idea 
resembles that in~\cite{AhmadianCohen:ICML2024}:
Queries where the set $L'$ is overrepresented (that is, there are more than $q |L'|$ keys from $L'$ in the query subset)
are on average more likely to get response ``1'' from queries in which $L'$ is underrepresented. These queries also have disproportionately more $L'$ keys than transparent keys and hence $L'$ keys on average are more likely to be scored.
    
\subsection{Proof of Theorem~\ref{metacorrectness:thm}}

We first establish that transparent keys all have around the same value in the count array $C$, and in particular, no transparent keys are placed in $M$.

For this section, our analysis of the algorithm will now span several steps, rather than just a single step. To this end, let $\eta_\phi^{(t)}$ be the failure probability of the QR in step $t$ (note that this is a random variable depending on $\phi$, which may itself depend on the outcome of previous steps.) Similarly define $p_\phi^{(t)}(x)$ and $p_\phi^{*(t)}$ to be the values of $p_\phi(x)$ and $p_\phi^{*}$, respectively, in step $t$. Also, let $C^{(t)}$ be the count array $C$ at time $t$.

Finally, for keys $x$ in the mask $M$, we modify the count array $C$ for the purposes of analysis by adding $p_\phi^{*(t)}$ to $C[x]$ in each step $t$ after $x$ was added to the mask $M$. Note that since $x$ is already in $M$, this does not affect anything in \cref{metaattack:algo}.

Also, let $E\ti$ be the total number of errors the QR has made until step $t$. Let $P\ti$ be the sum of $p_\phi^{*(j)}$ for all steps $j \le t$.

\begin{claim} \label{transparent-concentration:claim}
With probability $1-O(\inv n)$, for all transparent keys $x$, for all times $t$, $\absx{C\ti[x] - P\ti} \le 8\sqrt{r\log (rn)}$.
\end{claim}
\begin{proof}
Note that at step $t$, $C\ti[x] - P\ti$ increases by $C\ti[x] - C^{(t-1)}[x] - p_\phi^{*(t)}$. But, conditioned on everything before step $t$, $C\ti[x] - C^{(t-1)}[x]$ is a Bernoulli random variable with probability $p_\phi^{*(t)}$ (by definition of $p^*_\phi$), so $C\ti[x] - P\ti$ is a martingale with increments bounded in $[-1, 1]$. The statement then follows from the Azuma--Hoeffding inequality (and a union bound over $x$).
\end{proof}

\begin{cor} \label{transparent:cor}
No transparent key is ever added to $M$.
\end{cor}
\begin{proof}
This follows immediately from \cref{transparent-concentration:claim}, since a transparent key is added to $M$ when its count exceeds the median count by $16\sqrt{r\log (rn)}$, and the number of transparent keys is $n - |L| > n/2$.
\end{proof}

Now, we show that, as long as the QR does not make too many errors, $M$ will eventually have almost all of the determining pool. We show this using the score advantage lemma via a martingale argument that shows, essentially, that the total score advantage must eventually get large if the QR does not make too many errors.

\begin{claim} \label{qr-linear-error:claim}
With probability $1-O(\inv n)$, after $r$ steps (for some $r = O(|L|^2 \log n)$), either the QR has made $\eta r$ errors for a constant $\eta > 0$, or we will have $|L'| \le \ell_0$ (that is, the mask $M$ will contain all but $\ell_0$ elements of $L$). Moreover, as $\e \to 0$, we have $\eta \to 1/4$.
\end{claim}
\begin{proof}
For all times $t$ where $|L'| \ge \ell_0$, write
\begin{equation} \label{Y-def:eq}
Y_t \coloneqq \sum_{x \in L} C\ti[x] - |L|P\ti - c_1t + c_2E\ti,
\end{equation}
where $c_1$ and $c_2$ are the constants in the score advantage lemma (\cref{scoreadvantage:lemma}). If $|L'| < \ell_0$ at time $t$, set $Y_t = Y_{t-1}$. Then, we have, if $|L'| \ge \ell_0$ at time $t$, that
\[Y_t - Y_{t-1} = \sum_{x \in L} (C\ti[x] - C^{(t-1)}[x]) - |L|p_\phi^{*(t)} - c_1 + c_2 \mathbf{1}_{\x{error in step $t$}}. \]
Note that this is bounded by $O(|L|)$. Now, for $x$ already in the mask $M$, we have $C\ti[x] - C^{(t-1)}[x] = p_\phi^{*(t)}$, as mentioned earlier in this section. For $x \in L' = L \setminus M$, on the other hand, we have $C\ti[x] - C^{(t-1)}[x] \sim \Bern[p_\phi^{(t)}(x)]$ (conditioned on everything before step $t$). Also, the probability of an error in step $t$ is (by definition) $\eta_\phi^{(t)}$. Thus, we can write, conditioned on everything before step $t$, if $|L'| \ge \ell_0$, then
\begin{align*}
\E[Y_t - Y_{t-1}] &= |M| p_\phi^{*(t)} + \sum_{x \in L'} p_\phi^{(t)}(x) - |L| p_\phi^{*(t)} - c_1 +  c_2 \eta_\phi^{(t)} \\
&= |L'| \ov p_\phi^{(t)}(L') - |L'| p_\phi^{*(t)} - c_1 +  c_2 \eta_\phi^{(t)}.
\end{align*}
However, by the score advantage lemma (\cref{scoreadvantage:lemma}), this is nonnegative. Therefore, $Y_t$ is a supermartingale, with increments bounded by $O(|L|)$. Therefore, by Azuma's inequality, we have with probability $1-O(\inv n)$ that $Y_t \ge -8|L|  \sqrt{t \log(rn)}$. We apply this inequality at time $r$. If $|L'| < \ell_0$ at time $r$, then we are already done, so assume otherwise. In that case, \eqref{Y-def:eq} is valid, so using this inequality and rearranging, we have
\begin{equation*}
\sum_{x \in L} C\ts[x] - |L|P\ts \ge c_1r - c_2E\ts - 8|L| \sqrt{r \log(rn)}.
\end{equation*}
Pick an arbitrary constant $0 < \gamma < 1$ such that as $\e \to 0$, $\gamma \to 1$.
Then, if $E\ts > \gamma(c_1/c_2)r$, that means we have $\eta r$ errors, for $\eta = \gamma(c_1/c_2)$, so we are done (since $\eta \to 1/4$ as $\e \to 0$). Assuming otherwise for a contradiction, we then get
\begin{equation} \label{martingale-bd:eq}
\sum_{x \in L} (C\ts[x] - P\ts) \ge (1-\gamma)c_1r - 8|L| \sqrt{r \log(rn)}.
\end{equation}
Now, consider the state of the algorithm at time $t \le r$. By \cref{transparent-concentration:claim}, we have $\texttt{median}(C) \le P\ti + 8\sqrt{r\log (rn)}$. Thus, if $C\ti[x] - P\ti \ge 16\sqrt{r\log(rn)}$, then $x$ gets added to the mask. Once $x$ gets added to the mask, $C\ti[x] - P\ti$ stops increasing, since they increase by the same amount at each step. Therefore, noting that $C\ti[x] - P\ti$ increases by at most $1$ in each step, we always have $C\ti[x] - P\ti \le 16\sqrt{r\log(rn)} + 1$ for all $x \in L$. Therefore, plugging this into \eqref{martingale-bd:eq}, we have
\begin{equation*}
|L|(16\sqrt{r\log(rn)} + 1) \ge (1-\gamma)c_1r - 8|L| \sqrt{r \log(rn)}.
\end{equation*}
However, if we pick $r = C |L|^2 \log n$ for a sufficiently large constant $C$, the above is impossible (recalling that $|L| \le n$). Thus we have a contradiction, completing the proof of \cref{qr-linear-error:claim}.
\end{proof}


We now bootstrap \cref{qr-linear-error:claim} to show that there must be $\eta r$ errors (rather than the pool going below size $\ell_0$), by adding transparent keys into the determining pool:

\begin{claim}
With probability $1-O(\inv n)$, after $r$ steps (for some $r = O(|L|^2 \log n)$), the QR has made $\eta r$ errors for a constant $\eta > 0$ such that $\eta \to 1/4$ as $\e \to 0$.
\end{claim}
\begin{proof}
Let $\tw L$ be the union of $L$ with $\ell_0 + 1$ transparent keys. Note that $\tw L$ is still a $\delta$-determining pool, since determining pools are preserved under adding elements. Also, the algorithm only ever uses the size of $|L|$ for selecting $\qmin, q_1$, and is otherwise independent of the choice of $|L|$; note that the choice of $q_1$ that we use is still valid if $|L|$ increases by a constant. Therefore, all of our analysis still applies if $L$ were replaced by $\tw L$. 

Therefore, by \cref{metacorrectness:thm}, for some $r = O(|\tw L|^2 \log n) = O(|L|^2 \log n)$, after $r$ steps, the QR has either made $\eta r$ errors or the mask contains all but at most $\ell_0$ elements of $\tw L$. However, in the latter case, the mask would contain a transparent key, which would contradict \cref{transparent:cor} (which still holds for the original choice of $L$). Thus, the QR must have made $\Omega(r)$ errors, completing the claim.
\end{proof}

This completes the proof of \cref{metacorrectness:thm}.



\input{430-score-adv-proof}

%% file: 430-score-adv-proof.tex
\subsection{Proof of the score advantage lemma} \label{scoreadvantageproof:sec}

\begin{proof} [ Proof of Lemma~\ref{scoreadvantage:lemma}]
Throughout, we denote $\ell = |L'|.$ We will not specify $\ell_0$ explicitly; rather, we will assume throughout that $\ell$ is sufficiently large, including for the sake of asymptotic notation.

For an estimator map $\phi$, denote by $\overline{\phi}(m)$ the expected value of $\phi$ (the probability that the response is $1$) conditioned on $|U\cap L'| = m$. Equivalently, it is the expectation of $\phi(Y)$ over randomly chosen subsets of $L'$ of size $m$:
\[
\overline{\phi}(m)  := \E_{Y\sim \mathtt{Uniform}\left[\binom{L'}{m} \right]}[\phi(Y)]\ .
\]

We now express the scoring probability of a transparent key in terms of $\phi$.

\begin{align}
    p_\phi^* &= 
    \E_{U\sim \mathcal{Q}}[\phi(U\cap L') \cdot \mathbf{1}_{x\in U}] \nonumber\\
    &=
    \int_0^1 q\cdot \E_{W\sim \Bern[q]^{L'}} \left[ \phi(W) \right] d\nu(q) \nonumber \\
    &= \int_0^1 q\cdot \E_{m\sim \Binomial(\ell,q)}\left[ \overline{\phi}\left( m\right) \right] d\nu(q) \label{avescoreout:eq}
\end{align}

The second equality holds since the distribution of $W = U \cap L'$ is $\Bern[q]^{L'}$.  Then note that probability that for rate $q$, a transparent key $x\not\in L'$ is included in the subset remains $q$ even when conditioned on $W$.

We similarly express the average scoring probability of a key in $L'$ in terms of $\phi$. We have
for $x\in L'$
\begin{align*}
p_\phi(x) &=
\E_{U\sim \mathcal{Q}}[\phi(U\cap L') \cdot \mathbf{1}_{x\in U}] \\
&=
\int_0^1\E_{W\sim \Bern[q]^{L'}}[\phi(W) \cdot \mathbf{1}_{x\in W}] d\nu(q)\\
&= \int_0^1 \E_{m \sim \Binomial(\ell,q)}\left[\E_{W\sim \texttt{Uniform}\left[\binom{L'}{m}\right]}[\phi(W)\cdot \mathbf{1}_{x\in W}]  \right] d\nu(q)
\end{align*}
Now note that for $0\leq m\leq \ell$:
\begin{align*}
    \lefteqn{\E_{x\sim \mathtt{Uniform}[L']} \left[\E_{W\sim \texttt{Uniform}\left[\binom{L'}{m}\right]}[\phi(W)\cdot  \mathbf{1}_{x\in W}]\right]}   \\
    &= \E_{W\sim \texttt{Uniform}\left[\binom{L'}{m}\right]}\left[ \E_{x\sim \mathtt{Uniform}[L']}[\phi(W)\cdot \mathbf{1}_{x\in W}] \right]  \\
    &= \E_{W\sim \texttt{Uniform}\left[\binom{L'}{m}\right]} \left[\frac{m}{\ell} \cdot \phi(W) \right] \\
    &= \frac{m}{\ell} \cdot \E_{W\sim \texttt{Uniform}\left[\binom{L'}{m}\right]} [\phi(W)]\\
    &= \frac{m}{\ell} \overline{\phi}\left(m\right)
\end{align*}

Combining we obtain
\begin{align} 
\overline{p}_\phi(L') &:=
    \E_{x\sim \mathtt{Uniform}[L']} \left[ p_\phi(x)\right] \nonumber\\
    &= 
    \int_0^1 \E_{m \sim \Binomial(\ell,q)}\left[ \frac{m}{\ell} \overline{\phi}\left(m\right) \right] d\nu(q) \label{avescorein:eq}
\end{align}


We now express the score advantage gap using \eqref{avescorein:eq} and \eqref{avescoreout:eq}:
\begin{align*}
\overline{p}_\phi(L') -p^*_\phi 
&= \int_0^1 \E_{m \sim \Binomial(\ell,q)}\left[ \left(\frac{m}{\ell}-q\right)\cdot  \overline{\phi}\left(m\right) \right] \, d\nu(q) \\
&= \sum_{m=0}^{\ell} \overline{\phi}\left(m\right) \int_0^1\left(\frac{m}{\ell}-q\right)\cdot  \Pr[\Binomial[\ell,q]=m] \, d\nu(q) \\
&= \sum_{m=0}^{\ell} \overline{\phi}\left(m\right) \int_0^1\left(\frac{m}{\ell}-q\right)\cdot \binom{\ell}{m} q^{m}(1-q)^{\ell-m} \, d\nu(q) \\
&= C_\nu \s_{m=0}^\ell \ov \phi(m) \psi(m),
\end{align*}
where we have defined
\begin{equation}
\psi(m) = \int_0^1 f(q) \p{\f m\ell - q} \binom \ell m q^{m-1}(1-q)^{\ell-m-1} \, dq. \label{psi-def:eq}
\end{equation}

We defer the bounding of this integral to \cref{integrals:sec}. Then, by the bound on $\psi(m)$ derived in \cref{psi-bound:prop}, we continue to bound the score advantage. Define $\qmin' = \qmin + \unc$, $q_1' = q_1 - \unc$, $q_2' = q_2 + \unc$, and $\qmax' = \qmax - \unc$; by \cref{psi-bound:prop}, we then have
\begin{align}
\overline{p}_\phi(L') -p^*_\phi 
&= C_\nu \s_{m=0}^\ell \ov \phi(m) \psi(m) \nonumber \\
&= -\f{C_\nu r_1}{\ell(\ell+1)} \sum_{\qmin' \le m/\ell \le q_1'} \ov\phi(m) + \f{C_\nu r_2}{\ell(\ell+1)} \sum_{q_2' \le m/\ell \le \qmax'} \ov\phi(m) + O(\ell^{-4/3}) \nonumber \\
&= \f{C_\nu r_2(\qmax - q_2)}{\ell + 1} - \f{C_\nu r_1}{\ell(\ell+1)} \sum_{\qmin' \le m/\ell \le q_1'} \ov\phi(m) \nonumber\\ &\qquad - \f{C_\nu r_2}{\ell(\ell+1)} \sum_{q_2' \le m/\ell \le \qmax'} (1 - \ov\phi(m)) + O(\ell^{-4/3}) \label{score-adv-bound:eq}
\end{align}

In order to bound this in terms of the failure probability $\eta_\phi$, we now bound $\eta_\phi$ in terms of the values of $\phi$. By the definition of $\eta_\phi$ in \eqref{eta-def:eq}, we have
\begin{align*}
\eta_\phi = \E_{U \sim Q}[\phi(U \cap L') \cdot \mathbf{1}_{|U \cup M| < A}] + \E_{U \sim Q}[(1 - \phi(U \cap L')) \cdot \mathbf{1}_{|U \cup M| > B}].
\end{align*}
Now, conditioned on $q$, we have $|U \cup M| \sim |M| + \Bin(n - |M|, q)$. Since $|M| \le |L|$, we thus have (in distribution) that $\Bin(n, q) \le |U \cup M| \le |L| + \Bin(n, q)$. Thus, if $q < q_1$, by a Chernoff bound, we have with probability $1-e^{-\Omega(n)}$ that $|U \cup M| < |L| + q_1 n + o(1) < An$ (here we have used the definition of $q_1$ in \eqref{q-bd-1:eq}). Therefore, except with probability $e^{-\Omega(n)}$, we have $\mathbf{1}_{|U \cup M| < A} \ge \mathbf{1}_{q < q_1}$. 

Similarly, for $q > q_2$, we also have with probability $1-e^{-\Omega(n)}$ that $|U \cup M| > + q_2 n + o(1) > Bn$. Thus, except with probability $e^{-\Omega(n)}$, we also have $\mathbf{1}_{|U \cup M| > B} \ge \mathbf{1}_{q > q_2}$. Thus, we may bound $\eta_\phi$ as follows.
\begin{align}
\eta_\phi &\ge \E_{q \sim \nu, U \sim \Bern[q]^\Uu}[\phi(U \cap L') \cdot \mathbf{1}_{q < q_1}] + \E_{q \sim \nu, U \sim \Bern[q]^\Uu}[(1 - \phi(U \cap L')) \cdot \mathbf{1}_{q > q_2}] - e^{-\Omega(n)} \nonumber \\
&= \int_0^1 \mathbf{1}_{q < q_1} \E_{W \sim \Bern[q]^{L'}}[\phi(W)] \, d\nu(q) + \int_0^1 \mathbf{1}_{q > q_2} \E_{W \sim \Bern[q]^{L'}}[(1-\phi(W))] \, d\nu(q) - e^{-\Omega(n)} \nonumber \\
&= \int_0^1 \mathbf{1}_{q < q_1} \E_{m \sim \Bin(\ell, q)}[\ov \phi(m)] \, d\nu(q) + \int_0^1 \mathbf{1}_{q > q_2} \E_{m \sim \Bin(\ell, q)}[(1-\ov \phi(m))] \, d\nu(q) - e^{-\Omega(n)} \nonumber \\
&= \sum_{0 \le m \le \ell} \ov \phi(m) \int_0^1 \mathbf{1}_{q < q_1} \binom \ell m q^m (1-q)^{\ell-m} \, d\nu(q) \nonumber \\
&\qquad + \sum_{0 \le m \le \ell} (1 - \ov \phi(m)) \int_0^1 \mathbf{1}_{q > q_2} \binom \ell m q^m (1-q)^{\ell-m} \, d\nu(q)
- e^{-\Omega(n)} \nonumber \\
&\ge \sum_{\qmin' \le m/\ell \le q_1'} \ov \phi(m) \int_0^1 \mathbf{1}_{q < q_1} \binom \ell m q^m (1-q)^{\ell-m} \, d\nu(q) \nonumber \\
&\qquad + \sum_{q_2' \le m/\ell \le \qmax'} (1 - \ov \phi(m)) \int_0^1 \mathbf{1}_{q > q_2} \binom \ell m q^m (1-q)^{\ell-m} \, d\nu(q)
- e^{-\Omega(n)} \nonumber \\
&\ge \f {r_3} \ell \sum_{\qmin' \le m/\ell \le q_1'} \p{\f m\ell - \qmin} \ov \phi(m) + \f {r_3} \ell \sum_{q_2' \le m/\ell \le \qmax'} \p{\qmax - \f m\ell} (1 - \ov \phi(m)) - e^{-\Omega(n)}. \label{eta-bound:eq}
\end{align}
where in the last step we have used \cref{error-bound-integral:prop}.

Now, \eqref{score-adv-bound:eq} and \eqref{eta-bound:eq} are both bounds on the score advantage and failure probability, respectively, which are linear in the $\phi(m)$ (which are bounded between 0 and 1). To complete the proof of \cref{scoreadvantage:lemma}, we will now relate these two bounds.

\begin{lemma}
We have
\begin{equation} \label{01bd:eq}
\sum_{\qmin' \le m/\ell \le q_1'} \ov\phi(m) \le \sqrt{2\ell \sum_{\qmin' \le m/\ell \le q_1'}^{\vphantom{\ell}}  \p{\f m\ell - \qmin} \ov \phi(m)} + O(1),  
\end{equation}
and
\begin{equation} \label{23bd:eq}
\sum_{q_2' \le m/\ell \le \qmax'} (1 - \ov\phi(m)) \le \sqrt{2\ell \sum_{q_2' \le m/\ell \le \qmax'}^{\vphantom{\ell}} \p{\qmax - \f m\ell} (1 - \ov \phi(m))} + O(1).
\end{equation}
\end{lemma}
\begin{proof}
We start by proving \eqref{01bd:eq}. Let
\[k = \flrx{\sum_{\qmin' \le m/\ell \le q_1'} \ov\phi(m)}.\]
Note that the coefficients of $\ov\phi(m)$ in the right-hand side of \eqref{01bd:eq} are increasing in $m$. Thus, conditioned on $k$, the sum in the right-hand side of \eqref{01bd:eq} is minimized when the first $k$ values of $\ov \phi(m)$ are 1 and the rest are 0 (since $\phi(m) \in [0, 1]$ for all $m$). Therefore, we have
\begin{align*}
\sum_{\qmin' \le m/\ell \le q_1'}  \p{\f m\ell - \qmin} \ov \phi(m)
&\ge \sum_{\qmin' \le m/\ell \le \qmin' + k/\ell}  \p{\f m\ell - \qmin} \\
&\ge \f{1}{\ell} + \f{2}{\ell} + \dots + \f{k}{\ell} \\
&\ge \f{k^2}{2\ell}.
\end{align*}
The conclusion of \eqref{01bd:eq} follows. The proof of \eqref{23bd:eq} is virtually identical.
\end{proof}

Now, the bound in \eqref{eta-bound:eq} implies that the sums in the right-hand sides of \eqref{01bd:eq} and \eqref{23bd:eq} sum to at most $\f{\ell}{r_3} \eta_\phi + e^{-\Omega(n)}$. Therefore, we may write
\begin{gather*}
\sum_{\qmin' \le m/\ell \le q_1'} \ov\phi(m) \le \sqrt{\f{2\ell^2}{r_3} \eta_\phi \gamma_1} + O(1),  \\
\sum_{q_2' \le m/\ell \le \qmax'} (1 - \ov\phi(m)) \le \sqrt{\f{2\ell^2}{r_3} \eta_\phi \gamma_2} + O(1),
\end{gather*}
where $\gamma_1 + \gamma_2 = 1$.
(Note that the $e^{-\Omega(n)}$ term is subsumed by the $O(1)$.)

Finally, we bound the score advantage gap by substituting the above bounds into \eqref{score-adv-bound:eq}:
\begin{align*}
\overline{p}_\phi(L') -p^*_\phi 
&\ge \f{C_\nu r_2(\qmax - q_2)}{\ell + 1} - \f{C_\nu(r_1 \sqrt{\gamma_1} + r_2 \sqrt{\gamma_2})}{\ell(\ell+1)} \sqrt{\f{2\ell^2}{r_3} \eta_\phi} + O(\ell^{-4/3}) \\
&\ge \f{r_2(\qmax - q_2)}{\ell + 1} - \f{C_\nu\sqrt{(r_1^2+r_2^2)(\gamma_1+\gamma_2)}}{\ell(\ell+1)} \sqrt{\f{2\ell^2}{r_3} \eta_\phi} + O(\ell^{-4/3})\\
&= \f{C_\nu}{\ell} - \f{C_\nu \sqrt{2(r_1^2+r_2^2)/r_3}}{\ell} \cdot \sqrt{\eta_\phi} + O(\ell^{-4/3}).
\end{align*}
Now, picking any constant $c > 0$, we have $\sqrt{\eta_\phi} \le c + \eta_\phi / 4c$. Thus, we have
\begin{align*}
\overline{p}_\phi(L') -p^*_\phi 
&\ge C_\nu \p{1 - c\sqrt{2(r_1^2+r_2^2)r_3}} \f1\ell - \p{\f{C_\nu \sqrt{2(r_1^2+r_2^2)r_3}}{4c}} \f{n_\phi}{\ell} + O(\ell^{-4/3}).
\end{align*}
If we pick $c$ to be sufficiently small, the coefficient of $\f1\ell$ is positive, and \eqref{score-advantage:eq} thus follows.

Moreover, in the $\et$ regime, note that by the values of $r_1, r_2, r_3$ in \cref{integrals:sec}, we have $\sqrt{2(r_1^2+r_2^2)r_3} = 2 + \oet(1)$. For this regime, we pick $c = 1/4$. Then, in \eqref{score-advantage:eq}, we will have $c_1 = C_\nu(1/2 + \oet(1))$ and $c_2 = C_\nu(2 + \oet(1))$, so as $\et$, $c_1/c_2 \to 1/4$, as desired.

This completes the proof of \cref{scoreadvantage:lemma}.
\end{proof}

%% file: 2000-integrals.tex
\appendix

\section{Bounds for Section~\ref{scoreadvantageproof:sec}} \label{integrals:sec}

\newcommand{\Kl}{K_{\ell, m}}

In this section we bound the integrals of the form
\begin{equation} \label{K-def:eq}
\Kl(g) \coloneqq \int_0^1 g(q) \binom \ell m q^{m-1} (1-q)^{\ell - m - 1} \, dq.
\end{equation}
Throughout this section, we will assume, as in \cref{scoreadvantageproof:sec} that $\ell$ is sufficiently large. We also assume that any function $g$ that we plug into $K$ is set to zero outside $[0, 1]$. 

We also sometimes abuse notation by writing an expression in terms of $q$ in place of $g$; for example, $\Kl(q^2)$ will denote $\Kl(g)$ where $g(q) = q^2$.

\begin{lemma}
For bounded (by a constant) $g$, we have
\begin{equation*}
\Kl(g) = \int_{m/\ell - \ell^{-1/3}}^{m/\ell + \ell^{-1/3}} g(q) \binom \ell m q^{m-1} q^{\ell - m - 1} \, dq \pm e^{-\Omega(\ell^{1/3})}.
\end{equation*}
\end{lemma}
\begin{proof}
Note that we may write
\[\binom \ell m q^{m-1} q^{\ell - m - 1} = \f{\ell(\ell-1)}{m(\ell-m)} \Pr[\Bin(\ell-2, q) = m-1],\]
assuming for now that $m \neq 0, \ell$. Note that $(m-1)/(\ell - 2) = m/\ell + O(\ell^{-1})$. Thus, by a Chernoff bound, if $|q - m/\ell| > \ell^{-1/3}$, we have $\Pr[\Bin(\ell-2, q) = m-1] < e^{-\Omega(\ell^{1/3})}$. Therefore, for $|q - m/\ell| > \ell^{-1/3}$, the integrand in \eqref{K-def:eq} has magnitude most $e^{-\Omega(\ell^{1/3})}$ (since $g$ is bounded), and the conclusion follows.

In the special case $m = 0$, the integrand is just $g(q) q^{-1} (1-q)^{\ell - 1}$, which has magnitude at most $e^{-\Omega(\ell^{2/3})}$ for $q > 1/3$, and the $m = \ell$ case is similar, so we are done.
\end{proof}
Since $\Kl$ is linear, we have the following corollary.
\begin{cor} \label{K-approx:cor}
Suppose that $g, h$ are bounded and equal in the interval $[m/\ell - \unc, m/\ell + \unc]$. Then, 
\[|\Kl(g) - \Kl(h)| = \err.\]
\end{cor}

By explicit integration, we have the following fact.
\begin{fact} \label{K-poly:fact}
We have, as long as $m \neq \ell$,
\begin{align*}
\Kl(1) &= \f{\ell}{m(\ell-m)}, \\
\Kl(q) &= \f{1}{(\ell-m)}, \\
\Kl(q^2) &= \f{m + 1}{(\ell+1)(\ell-m)}.
\end{align*}
(The first of these also requires $m \neq 0$.)
\end{fact}
\begin{proof}
For $i = 0, 1, 2$,
\begin{align*}
\Kl(q^i) &= \int_0^1 \binom \ell m q^{m+i-1} (1-q)^{\ell - m - 1} \, dq \\
&= B(m+i, \ell-m) \binom \ell m \\
&= \f{(m+i-1)!(\ell-m-1)!}{(\ell+i-1)!} \binom \ell m,
\end{align*}
and the conclusion follows. (Here $B(a, b)$ is the beta function.)
\end{proof}

Linearity of $\Kl$ implies the following corollary.
\begin{cor} \label{K-linear:cor}
We have, for all constants $a, b$,
\begin{equation*}
\Kl((aq+b)(m/\ell - q)) = -\f{a}{\ell(\ell+1)}.
\end{equation*}
\end{cor}
\begin{proof}
For $m \neq \ell$, this follows directly from \cref{K-poly:fact} (we have to handle the $m=0$ case specially, but note that the constant coefficient of $(aq+b)(m/\ell - q)$ is zero in this case). For $m = \ell$, the conclusion follows from the $m=0$ case and substituting $m \mapsto \ell-m$ and $q \mapsto 1-q$.
\end{proof}

We also show the following.
\begin{lemma}
Suppose that $\tw g(q) = g(q)(m/\ell - q)$, for an $L$-Lipschitz bounded function $g$ (where $L \ge 0$ is a constant). Then,
\begin{equation*}
|\Kl(\tw g)| = O(1/\ell^2).
\end{equation*}
\end{lemma}
\begin{proof}
Consider the function $h(q) = L(q - m/\ell) + g(m/\ell)$, which is the linear function with slope $L$ that matches $g$ at $m/\ell$. Note that $h(q) \ge g(q)$ when $q \ge m/\ell$, and $h(q) \le g(q)$ when $q \le m/\ell$ (since $g$ is $L$-Lipschitz). Thus, letting $\tw h(q) = h(q) (m/\ell - q)$, we have $\tw h(q) \le \tw g(q)$ for all $q$. Thus, applying \cref{K-linear:cor},
\[\Kl(\tw g) \ge \Kl(\tw h) = -\f{L}{\ell(\ell+1)}.\]
If we similarly do this for a function with slope $-L$, we get
\[\Kl(\tw g) \le \f{L}{\ell(\ell+1)}.\]
This completes the proof of the lemma.
\end{proof}

\begin{prop} \label{psi-bound:prop}
We have, for the constants $r_1 = 1/(q_1-\qmin)$ and $r_2 = 1/(\qmax - q_2)$,
\begin{equation} \label{psi-expr:eq}
\psi(m) = 
\begin{cases}
-r_1/\ell(\ell+1), & m/\ell \in [\qmin + \unc, q_1 - \unc] \\
r_2/\ell(\ell+1), & m/\ell \in [q_2 + \unc, \qmax - \unc] \\
O(1/\ell^2), & m/\ell \in [q_i - \unc, q_i + \unc] \textnormal{ for } i \in \{\textnormal{min}, 1, 2, \textnormal{max}\} \\
0, & \textnormal{else}
\end{cases}
\pm \err. 
\end{equation}
\end{prop}
\begin{proof}
Recall the definition of $\psi(m)$ in \eqref{psi-def:eq}: we have $\psi(m) = \Kl(\tw f)$, where  $\tw f(q) = f(q)(m/\ell - q)$.

Also recalling the definition of $f$ in \cref{f-def:eq}, $f$ is $O(1)$-Lipschitz, so $\psi(m) = O(1/\ell^2)$ for all $m$. For all $m$ which are at least $\unc$ away from some $q_i$, we can, by \cref{K-approx:cor}, replace $f$ by the corresponding linear function which determines $f$ on the interval that $m/\ell$ is in. The result \eqref{psi-expr:eq} then follows from \cref{K-linear:cor}.
\end{proof}

\begin{prop} \label{error-bound-integral:prop}
There is a constant $r_3$ such that for $\qmin + \unc \le m/\ell \le q_1 - \unc$, we have
\begin{equation} \label{int1:eq}
\int_0^1 \mathbf{1}_{q < q_1} \binom \ell m q^m (1-q)^{\ell-m} \, d\nu(q) \ge \f{r_3}{\ell} \p{\f m\ell - \qmin} - \err, 
\end{equation}
and for $q_2 + \unc \le m/\ell \le \qmax - \unc$, we have
\begin{equation} \label{int2:eq}
\int_0^1 \mathbf{1}_{q > q_2} \binom \ell m q^m (1-q)^{\ell-m} \, d\nu(q) \ge \f{r_3}{\ell} \p{\qmax - \f m\ell} - \err.
\end{equation}
Moreover, in the regime where $\et$, we have $r_3 = (1+\oet(1))/\e$.
\end{prop}

\begin{proof}
We will show \eqref{int1:eq}; the proof of \eqref{int2:eq} is essentially identical.

Plugging in the definition of $\nu$ (in \cref{ratedist:def}), we have that the left-hand side of \eqref{int1:eq} is
\begin{equation*}
C_\nu \int_0^1 f(q) \mathbf{1}_{q < q_1} \binom \ell m q^{m-1} (1-q)^{\ell-m-1} \, dq
= C_\nu \Kl(f(q)\mathbf{1}_{q < q_1}).
\end{equation*}
Now, note that in the given range of $m$, for all $q \in [m/\ell - \unc, m/\ell + \unc]$, we have $f(q)\mathbf{1}_{q < q_1} = \f{q-\qmin}{q_1-\qmin}$. Thus, by \cref{K-approx:cor} (and then \cref{K-poly:fact}), we have
\begin{align*}
\Kl(f(q)\mathbf{1}_{q < q_1})
&\ge \Kl\paf{q-\qmin}{q_1-\qmin} - \err \\
&= \f{1}{q_1-\qmin} \p{\f{1}{\ell-m} - \qmin \cdot \f{\ell}{m(\ell-m)}} - \err \\
&= \f{m/\ell-\qmin}{(q_1-\qmin)(m/\ell)(1-m/\ell)\ell} - \err.
\end{align*}
Since $m/\ell$ is bounded away from both 0 and 1 by a constant, the first conclusion follows. Moreover, in the $\e \to 0$ regime, we have $m/\ell = (1+\oet(1)) q_1$, so it follows that
\begin{align*}
r_3 &= \f{(1+\oet(1)) C_\nu}{(q_1-\qmin) q_1 (1-q_1)} \\
&= \f{1+\oet(1)}{\e},
\end{align*}
where we have plugged in the value of $C_\nu$ from \eqref{eq:cnu}. Thus, we are done.
\end{proof}